\documentclass{lmcs}
\pdfoutput=1

\usepackage{lastpage}
\lmcsdoi{20}{3}{4}
\lmcsheading{}{\pageref{LastPage}}{}{}%
{Mar.~24,~2023}{Jul.~08,~2024}{}

\usepackage{amssymb}
\usepackage{hyperref}
\usepackage{cleveref}
\usepackage[utf8]{inputenc}
\usepackage[T1]{fontenc}
\usepackage{mathrsfs}
\usepackage{mathtools}
\usepackage[all]{xy}
\usepackage[defaultlines=3,all]{nowidow}
\usepackage{graphicx}
\usetikzlibrary{fit}
\usepackage{nicematrix}

\DeclareMathOperator{\tww}{tww}

\renewcommand{\phi}{\varphi}

\newcommand{\cyrsffamily}{\fontencoding{OT2}\fontfamily{cmss}\selectfont}
\DeclareTextFontCommand{\textcyrsf}{\cyrsffamily}

\DeclareMathOperator{\sh}{\mathsf{Reduct}}

\DeclareMathOperator{\bl}{bf}
\crefname{thm}{Theorem}{Theorems}
\crefname{thmC}{Theorem}{Theorems}
\crefname{lem}{Lemma}{Lemmas}

\newtheorem*{thm*}{Theorem}
\newtheoremstyle{claim}
{6pt}
{6pt}
{\normalfont}
{0pt}
{\bfseries}
{{\bfseries .}}
{5pt plus 1pt minus 1pt}
{$\vartriangleright$ \thmname{#1}\thmnumber{ #2}\textnormal{\thmnote{ (#3)}}}
\theoremstyle{claim}
\newtheorem{claim}[thm]{Claim}
\newenvironment{claimproof}{\begin{proof}[Proof of the claim.]}{\end{proof}}
\newenvironment{absolutelynopagebreak}
{\par\nobreak\vfil\penalty0\vfilneg
	\vtop\bgroup}
{\par\xdef\tpd{\the\prevdepth}\egroup
	\prevdepth=\tpd}

		\crefname{thm}{thm}{Theorems}
		\crefname{defi}{Definition}{Definitions}
		\crefname{lem}{lem}{Lemmas}
		\crefname{cor}{Corollary}{Corollaries}
		\crefname{conj}{Conjecture}{Conjectures}
		\crefname{exa}{Example}{Examples}
		\crefname{prob}{Problem}{Problems}
		\crefname{claim}{Claim}{Claims}
		\crefname{prop}{Proposition}{Propositions}
		\crefname{fact}{Fact}{Facts}
		
\title{Twin-width and permutations} 

\thanks{\hspace*{-12pt}\begin{minipage}[c]{0.8\textwidth}\hspace*{12pt}This paper is part of a project that has received funding from the European Research Council (ERC) under the European Union's Horizon 2020 research and innovation program (grant agreement No 810115 -- {\sc Dynasnet}) and from the German Research Foundation (DFG) with grant agreement
No 444419611.\end{minipage}\hfill
\begin{minipage}[c]{.16\textwidth}
	\includegraphics[height=1cm]{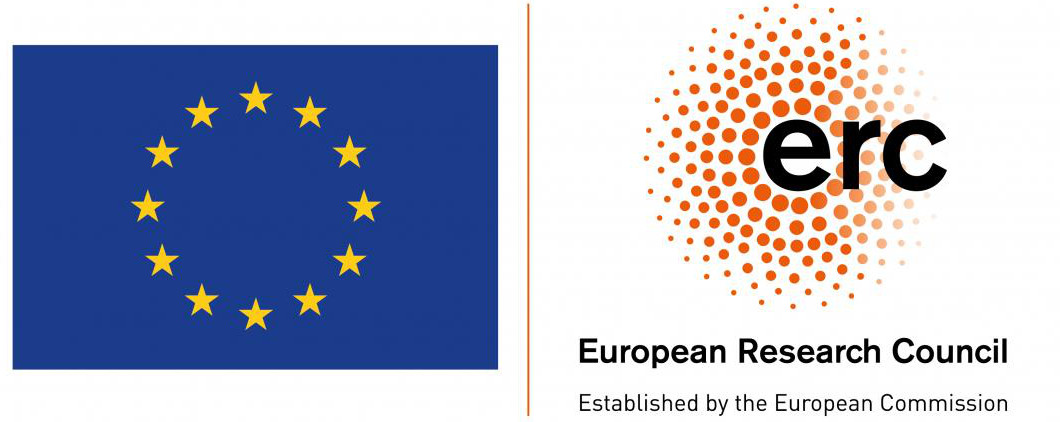}\\
	\includegraphics[height=1cm]{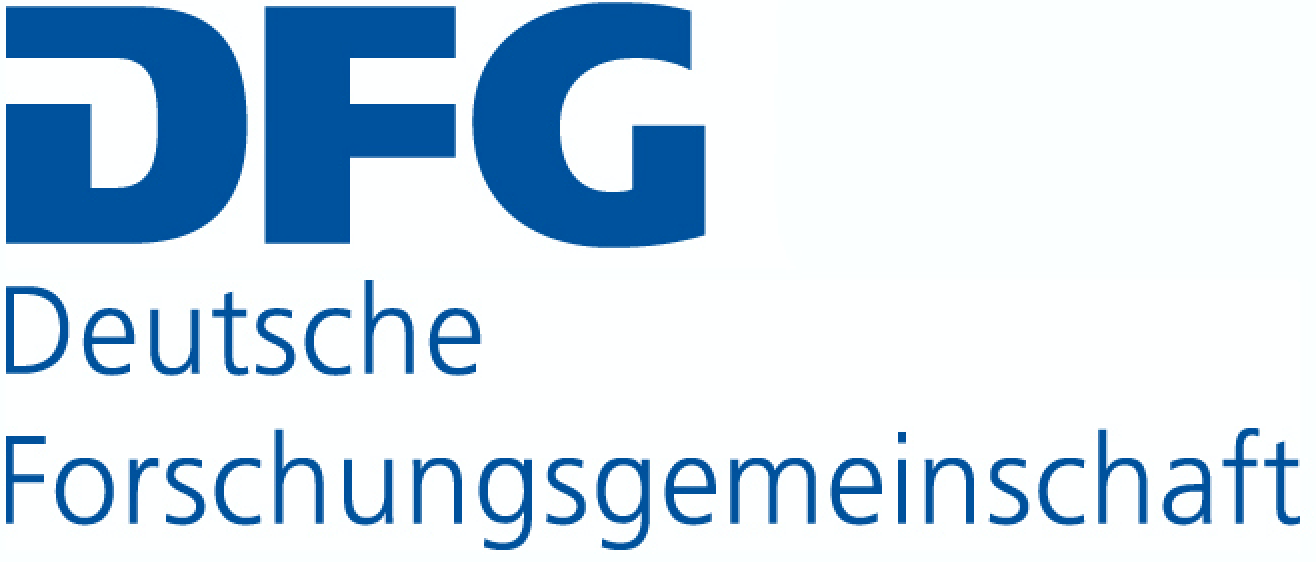}
\end{minipage}{\color{white}.}}

\author[\'E.~Bonnet]{\'Edouard Bonnet\lmcsorcid{0000-0002-1653-5822}}[a]
\author[J.~Ne\v set\v ril]{Jaroslav Ne\v set\v ril\lmcsorcid{0000-0002-5133-5586}}[b]
\author[P.~Ossona de Mendez]{Patrice {Ossona de Mendez}\lmcsorcid{0000-0003-0724-3729}}[c,b]
\author[S.~Siebertz]{Sebastian Siebertz\lmcsorcid{0000-0002-6347-1198}}[d]
\author[S.~Thomass{\'e}]{St\'ephan Thomass{\'e}\lmcsorcid{0000-0002-7090-1790}}[a]

\address{Universit\'e de  Lyon, CNRS, ENS de Lyon, France}
\email{edouard.bonnet@ens-lyon.fr, stephan.thomasse@ens-lyon.fr}

\address{Computer Science Institute of Charles University (IUUK), Praha, Czech Republic}
\email{nesetril@iuuk.mff.cuni.cz}

\address{Centre d'Analyse et de Math\'ematiques Sociales, EHESS, CNRS UMR 8557, Paris, France}
\email{pom@ehess.fr}

\address{University of Bremen, Bremen, Germany}
\email{siebertz@uni-bremen.de}


\begin{document}
	\begin{abstract}
		Inspired by a width invariant on permutations  defined by Guillemot and
		Marx, Bonnet, Kim, Thomass\'e, and Watrigant introduced the  twin-width of graphs, which 
		is a parameter describing its structural complexity.
		This invariant has been further extended to binary structures, in several (basically equivalent) ways.
		We prove that a class of binary relational structures (that is: edge-colored partially directed graphs) has bounded twin-width if and only if it is a first-order transduction of a~proper permutation class.
		As a by-product, we show that every class with bounded twin-width contains at most $2^{O(n)}$ pairwise non-isomorphic $n$-vertex graphs.
	\end{abstract}
	
	\maketitle
	

\section{Introduction}

In this paper we consider the graph parameter \emph{twin-width}, defined by
Bonnet, Kim, Thomass\'e and Watrigant~\cite{twin-width1}
as a generalization of an invariant for classes of permutations
defined by Guillemot and Marx~\cite{Guillemot14}. Twin-width was
recently studied intensively in the context of many structural and
algorithmic questions, such as FPT model checking~\cite{twin-width1},
graph enumeration~\cite{twin-width2}, graph
coloring~\cite{twin-width3}, and structural properties of matrices and ordered
graphs~\cite{twin-width4}.  

Many well-studied classes of graphs have bounded twin-width: planar graphs, and more generally, any class of graphs excluding a fixed minor, cographs, and more generally, any class of bounded clique-width, etc.

The twin-width of graphs was originally defined using a sequence of `near-twin' vertex contractions or identifications.
Roughly speaking, twin-width measures the accumulated error (recorded via the so-called `red edges') made by the identifications.
To help the reader start forming intuitions, we give a concise  definition of the twin-width of a graph; a formal generalization for binary structures is presented in~\Cref{sec:tww}.

A~\emph{trigraph} is a graph with some edges colored red (while the rest of them are black).
A~\emph{contraction} (or identification) consists of merging two (non-necessarily adjacent) vertices, say, $u, v$ into a vertex $w$ that is adjacent to a vertex $z$ via a black edge if $uz$ and $vz$ were black edges, or otherwise, via a red edge if at least one of $u$ and $v$ were adjacent to $z$. The rest of the trigraph does not change.
A~\emph{contraction sequence} of an $n$-vertex graph $G$ is a sequence of trigraphs $G=G_n, \ldots, G_1$ such that $G_i$ is obtained from $G_{i+1}$ by performing one contraction (observe that $G_1$ is the 1-vertex graph).
A~\emph{$d$-sequence} is a contraction sequence where all the trigraphs have red degree at most~$d$.
The~\emph{twin-width} of $G$ is then the minimum integer~$d$ such that $G$ admits a $d$-sequence.
See~\Cref{fig:contraction-sequence} for an example of a graph admitting a 2-sequence.

\begin{figure}[h!t]
	\centering
	\resizebox{400pt}{!}{
		\begin{tikzpicture}[
			vertex/.style={circle, draw, minimum size=0.68cm}
			]
			\def\s{1.2}
			\foreach \i/\j/\l in {0/0/a,0/1/b,0/2/c,1/0/d,1/1/e,1/2/f,2/1/g}{
				\node[vertex] (\l) at (\i * \s,\j * \s) {$\l$} ;
			}
			\foreach \i/\j in {a/b,a/d,a/f,b/c,b/d,b/e,b/f,c/e,c/f,d/e,d/g,e/g,f/g}{
				\draw (\i) -- (\j) ;
			}
			
			\begin{scope}[xshift=3 * \s cm]
				\foreach \i/\j/\l in {0/0/a,0/1/b,0/2/c,1/0/d,2/1/g}{
					\node[vertex] (\l) at (\i * \s,\j * \s) {$\l$} ;
				}
				\foreach \i/\j/\l in {1/1/e,1/2/f}{
					\node[vertex,opacity=0.2] (\l) at (\i * \s,\j * \s) {$\l$} ;
				}
				\node[draw,rounded corners,inner sep=0.01cm,fit=(e) (f)] (ef) {ef} ;
				\foreach \i/\j in {a/b,a/d,b/c,b/d,b/ef,c/ef,c/ef,d/g,ef/g,ef/g}{
					\draw (\i) -- (\j) ;
				}
				\foreach \i/\j in {a/ef,d/ef}{
					\draw[red, very thick] (\i) -- (\j) ;
				}
			\end{scope}
			
			\begin{scope}[xshift=6 * \s cm]
				\foreach \i/\j/\l in {0/1/b,0/2/c,2/1/g,1/1/ef}{
					\node[vertex] (\l) at (\i * \s,\j * \s) {$\l$} ;
				}
				\foreach \i/\j/\l in {0/0/a,1/0/d}{
					\node[vertex,opacity=0.2] (\l) at (\i * \s,\j * \s) {$\l$} ;
				}
				\draw[opacity=0.2] (a) -- (d) ;
				\node[draw,rounded corners,inner sep=0.01cm,fit=(a) (d)] (ad) {ad} ;
				\foreach \i/\j in {ad/b,b/c,b/ad,b/ef,c/ef,c/ef,ef/g,ef/g}{
					\draw (\i) -- (\j) ;
				}
				\foreach \i/\j in {ad/ef,ad/g}{
					\draw[red, very thick] (\i) -- (\j) ;
				}
			\end{scope}
			
			\begin{scope}[xshift=9 * \s cm]
				\foreach \i/\j/\l in {0/2/c,2/1/g,0.5/0/ad}{
					\node[vertex] (\l) at (\i * \s,\j * \s) {$\l$} ;
				}
				\foreach \i/\j/\l in {0/1/b,1/1/ef}{
					\node[vertex,opacity=0.2] (\l) at (\i * \s,\j * \s) {$\l$} ;
				}
				\draw[opacity=0.2] (b) -- (ef) ;
				\node[draw,rounded corners,inner sep=0.01cm,fit=(b) (ef)] (bef) {bef} ;
				\foreach \i/\j in {ad/bef,bef/c,bef/ad,c/bef,c/bef,bef/g}{
					\draw (\i) -- (\j) ;
				}
				\foreach \i/\j in {ad/bef,ad/g,bef/g}{
					\draw[red, very thick] (\i) -- (\j) ;
				}
			\end{scope}
			
			\begin{scope}[xshift=11.7 * \s cm]
				\foreach \i/\j/\l in {0/2/c}{
					\node[vertex] (\l) at (\i * \s,\j * \s) {$\l$} ;
				}
				\foreach \i/\j/\l in {0.5/0/adg,0.5/1.1/bef}{
					\node[vertex] (\l) at (\i * \s,\j * \s) {\footnotesize{\l}} ;
				}
				\foreach \i/\j in {c/bef}{
					\draw (\i) -- (\j) ;
				}
				\foreach \i/\j in {adg/bef}{
					\draw[red, very thick] (\i) -- (\j) ;
				}
			\end{scope}
			
			\begin{scope}[xshift=13.7 * \s cm]
				\foreach \i/\j/\l in {0.5/0/adg,0.5/1.1/bcef}{
					\node[vertex] (\l) at (\i * \s,\j * \s) {\footnotesize{\l}} ;
				}
				\foreach \i/\j in {adg/bcef}{
					\draw[red, very thick] (\i) -- (\j) ;
				}
			\end{scope}
			
			\begin{scope}[xshift=15 * \s cm]
				\foreach \i/\j/\l in {1/0.75/abcdefg}{
					\node[vertex] (\l) at (\i * \s,\j * \s) {\tiny{\l}} ;
				}
			\end{scope}
			
		\end{tikzpicture}
	}
	\caption{A 2-sequence witnessing that the initial graph has twin-width at most~2.}
	\label{fig:contraction-sequence}
\end{figure}

In this paper, the extension of twin-width for binary relational structures perfectly matches the one in~\cite{twin-width1} on undirected graphs, but will slightly differ for general binary structures.
Though, as we will observe, both definitions give parameters that differ only by, at most, a linear factor. 

We show that twin-width can be concisely expressed by special structures, which we call \emph{twin-models}.
Twin-models are rooted trees augmented by a set of transversal edges that
satisfies two simple properties: minimality and consistency. These
properties imply that every twin-model admits a \emph{ranking}, from which we
can compute a \emph{width}. The twin-width of a structure then coincides with the
optimal width of a ranked twin-model of the structure.  While this
connection is technical, twin-models provide a simple way to handle
classes of binary structures with bounded twin-width.
Note that an informal precursor of ranked twin-models appears in~\cite{twin-width3} in the form of the so-called \emph{ordered union trees} and the realization that the edge set of graphs of twin-width at most~$d$ can be partitioned into~$O_d(n)$ bicliques where both sides of each biclique are a discrete interval along a unique fixed vertex ordering.
The main novelty in the (ranked) twin-models lies in the axiomatization of \emph{legal} sets of transversal edges, which is indispensable to their logical treatment.

\smallskip
This paper is a combination of model-theoretic tools (relational
structures, interpretations, transductions), structural graph theory
and theory of permutations. Here, by a permutation, we mean a relational structure consisting of two linear orders on the same set (see \cite{albert2020two} for a discussion on representations of permutations). Note that this type of representation is particularly adapted to the study of patterns in permutations. 
The following is the main result of this paper:

\vspace{-2mm}

\begin{thm*}
	A class of binary relational structures has bounded twin-width if and only if it is a first-order transduction of a proper permutation class.
\end{thm*}

The ``only if'' part of this theorem is stated in more technical terms in \Cref{sec:perm} as \Cref{thm:main}, and is our main contribution.
The other direction, the fact that every binary structure that is a~first-order transduction of a~proper permutation class has bounded twin-width, was already known~\cite{twin-width1}.
More specifically, it was shown that proper permutation classes have bounded twin-width~\cite[Section 6.1]{twin-width1} and that every first-order transduction of a~class of bounded twin-width has itself bounded twin-width~\cite[Theorem 8.1]{twin-width1}. 

We recall that a proper permutation class is a set of permutations closed under sub-permutations that excludes at least one permutation. 
Transductions provide a model theoretical tool to encode relational
structures (or classes of relational structures) inside other 
(classes of) relational structures and will be formally defined in \Cref{sec:transductions}. 

The fact that any class of graphs with
bounded twin-width is just a transduction of a very simple class (a proper permutation
class) is surprising at first glance, and it nicely complements another
model theoretic characterization of classes of bounded twin-width: a class of graphs has bounded twin-width if and only if it is the reduct of a  dependent class of ordered graphs \cite{twin-width4}.
It can also be thought of as scaling up the fact that classes of bounded rank-width coincide with transductions of tree orders, and classes of bounded linear rank-width, with transductions of linear orders~\cite{Colcombet07}.
On the other hand,
twin-models are interesting objects per se and in a way present one of the
most permissive forms of width parameters related to trees. 
Note that for other classes of sparse structures we do not have such concrete models.

The main result implies that every relational structure on~$n$ elements
from a class with bounded twin-width can be encoded in a permutation on at most 
$kn$ elements for some number $k$. 
It is then a consequence of~\cite{MarcusT04} that
every class of relational structures with bounded twin-width contains at
most $c^n$ non-isomorphic structures with $n$ vertices, hence is small
(i.e., contains at most $c^n\,n!$ labeled structures with $n$
elements).
This extends the main result of~\cite{twin-width2} while not using the ``versatile twin-width'' machinery  (but only the preservation of bounded twin-width by transductions proved in~\cite{twin-width1}). 
This also extends a similar property for proper minor-closed
classes of graphs, which can be derived from the boundedness of book thickness, as noticed by McDiarmid (see the concluding remarks of \cite{BERNARDI2010468}). 

The proof of our main result is surprisingly complex and proceeds in several steps, which perhaps add new aspects to the rich spectrum of structures related to twin-width. The basic  steps can be outlined as follows (the relevant terminology will be formally introduced in the appropriate sections).

We start with a class $\mathscr C_0$ of binary relational structures
with bounded twin-width. We derive a class
$\mathscr T$ of twin-models (tree-like representations of the structures using rooted binary trees and transversal binary relations). 
Replacing the rooted binary trees of
the twin-models by binary tree orders, we get a class $\mathscr F$
of so-called full twin-models, which we prove has bounded twin-width. This class can
be used to retrieve $\mathscr C_0$ as a transduction, that is by means of a logical encoding.  Using a
transduction pairing (generalizing the notion of a bijective encoding)
between binary tree orders~$\mathscr O$ and
rooted binary trees ordered by a preorder~$\mathscr Y^<$ we
derive a transduction pairing of the class of full twin-models $\mathscr F$ with a class~$\mathscr T^<$ of ordered twin-models. From the property that the
class $\mathscr G$ of the Gaifman graphs of the twin-models in~$\mathscr T$ is degenerate (and has bounded twin-width), we prove a
transduction pairing of~$\mathscr T$ and $\mathscr G$, from
which we derive a transduction pairing of $\mathscr T^<$ and
the class~$\mathscr G^<$ of ordered Gaifman graphs of the ordered
twin-models. As a composition of a transduction pairing of $\mathscr G^<$ with a
class $\mathscr E^<$ of ordered binary structures, in which each binary
relation induces a pseudoforest and a transduction pairing of
$\mathscr E^<$ with a class~$\mathscr P$ of permutations we define a transduction pairing of $\mathscr G^<$ and $\mathscr P$. As $\mathscr G^<$ has bounded twin-width (as it is a transduction of a class
with bounded twin-width) we infer that $\mathscr P$ avoids at least one pattern. Following the backward transductions, we
eventually deduce that $\mathscr C_0$ is a transduction of
the hereditary closure $\overline{\mathscr P}$ of $\mathscr P$, which is a proper permutation class.

This proof may be schematically outlined by \Cref{fig:diag}. Here, all the notations are consistent with the notation used later in our proof.

\newcommand{\rdef}[1]{{\scriptsize (Def.~\ref{#1})}}
\newlength{\blen}
\setlength{\blen}{25pt}

\begin{figure}[h!t]
	\centering
	\begin{minipage}{\textwidth}\renewcommand{\thefootnote}{\alph{footnote}}
		\[
		\xy
		\xymatrix@C=.15\textwidth@R=50pt{
			&\makebox[\blen][c]{\parbox{35pt}{\centering $\mathscr{O}_0$\\\rdef{def:clOY}}}\ar@{->>}@/^/@<8pt>[r]^{\mathsf{L}}\ar@{}@<8pt>[r]|{\rightleftharpoons}&
			\makebox[\blen][c]{\parbox{35pt}{\centering $\mathscr{Y}_0^<$\\\rdef{def:clOY}}}\ar@{->>}@/^/@<-8pt>[l]^{\mathsf{O}}\\
			&\makebox[\blen][c]{\parbox{35pt}{\centering$\mathscr F$\\{\rdef{def:tS}}}}\ar@{->>}[dl]_{\parbox{65pt}{\scriptsize\centering ($2$-bounded) $\mathsf{S}$}}\ar@{->>}@/^/@<8pt>[r]^{\widehat{\mathsf L}}\ar@{}@<8pt>[r]|{\rightleftharpoons}\ar@{->>}[u]_{\mathsf{Reduct}}&\makebox[\blen][c]{\parbox{35pt}{\centering $\mathscr T^<$\\\rdef{def:clOY}}}\ar@{->>}@/^/@<-8pt>[l]^{\widehat{\mathsf O}}\ar@{->>}@/^/@<8pt>[r]^{\widehat{\mathsf{G}}}\ar@{->>}[d]^{\mathsf{Reduct}}\ar@{->>}[u]_{\mathsf{Reduct}}\ar@{}@<8pt>[r]|{\rightleftharpoons}&\makebox[\blen][c]{\parbox{35pt}{\centering $\mathscr G^<$\\\rdef{def:clGo}}}\ar@{->>}@/^/@<-8pt>[l]^{\widehat{\mathsf U}}\ar@{->>}[d]^{\mathsf{Reduct}}\ar@{->>}@/^/@<8pt>[r]^{\mathsf{T}_1}\ar@{}@<8pt>[r]|{\rightleftharpoons}&\parbox{40pt}{\centering$\mathscr P\subseteq\overline{\mathscr P}$\\{\scriptsize~}}\ar@{->>}@/^/@<-8pt>[l]^{\mathsf{T}_2}\\
			\parbox{11pt}{\centering $\mathscr C_0$\\{\scriptsize~}}\ar@{..>}@<5pt>[rr]_{\parbox{65pt}{\scriptsize\centering twin-model}}&&\makebox[\blen][c]{\parbox{35pt}{\centering$\mathscr T$\\\rdef{def:clT}}}
			\labelmargin-{6pt}
			\ar@{..>}@/^/[ul]^>>>>>>>{\scriptsize\parbox{50pt}{\centering full\\ twin-model}}
			\labelmargin+{6pt}
			\ar@{->>}@/^/@<8pt>[r]^{\mathsf{G}}\ar@{}@<8pt>[r]|{\rightleftharpoons}&\makebox[\blen][c]{\parbox{35pt}{\centering $\mathscr G$\\\rdef{def:clG}}}\ar@{->>}@/^/@<-8pt>[l]^{\mathsf{U}}
		}
		\endxy
		\]
	\end{minipage}
	\caption{Relations between the classes of structures involved in the proof of the main result.
	The interpretation $\mathsf S$ is defined in \Cref{def:tS}, 
	the transduction pairing $(\mathsf L,\mathsf O)$ in \Cref{lem:otm},
the  transduction pairing $(\widehat{\mathsf L},\widehat{\mathsf O})$ as a remark just after \Cref{def:clOY}, 
the transduction pairing $(\mathsf G,\mathsf U)$ in \Cref{lem:toGaifman}, 
and the  transduction pairing $(\widehat{\mathsf G},\widehat{\mathsf U})$ as a remark just after \Cref{def:clGo}.
		\label{fig:diag}
	}
\end{figure}

The full transformation of a graph $G$ into a permutation $\sigma$ and the inverse transformation (obtained as a transduction) are displayed on \Cref{fig:full} on an example.

\clearpage
\begin{figure}[t]
	\centering
	\includegraphics[width=\textwidth]{full}
	\caption{From a graph $G$ to a permutation $\sigma$, and back.}
	\label{fig:full}
\end{figure}

\section{Preliminaries}
\subsection{Relational structures}
\label{sec:rel}
We assume basic knowledge of first-order logic and refer
to~\cite{hodges1997shorter} for extensive background.
A~\emph{relational signature} $\Sigma$ is a finite set of relation
symbols $R_i$ with associated arity $r_i$. A~\emph{relational
	structure}~$\mathbf A$ with signature $\Sigma$, or simply a
\emph{$\Sigma$-structure} consists of a \emph{domain}~$A$ together
with relations $R_i({\mathbf A})\subseteq A^{r_i}$ for each relation
symbol $R_i\in\Sigma$ with arity $r_i$. The
relation~$R_i({\mathbf A})$ is called the \emph{interpretation}
of~$R_i$ in~$\mathbf A$.  We will often speak of a relation instead of
a relation symbol when there is no ambiguity. We may write $\mathbf A$
as $(A,R_1({\mathbf A}),\dots,R_s({\mathbf A}))$. In this paper we
will consider relational structures with finite
We will further assume
that $\Sigma$-structures are \emph{irreflexive}, that is,
$(v,v)\not\in R_i(\mathbf{A})$ for every element $v\in A$ and relation
symbol $R_i\in\Sigma$. A unary relation is called a \emph{mark}. Let
$R$ be a binary relation symbol and let $u,v\in A$.  That the pair
$(u,v)$ lies in the interpretation of $R$ in $\mathbf A$ will be
indifferently denoted by $(u,v)\in R(\mathbf A)$ or 
$\mathbf A\models R(u,v)$.
More generally, for a formula $\phi(x_1,\dots,x_k)$, a
$\Sigma$-structure~$\mathbf A$, an integer $\ell<k$ and
$a_1,\dots,a_\ell\in A$ we define
$\phi(\mathbf A,a_1,\dots,a_\ell):= \{(x_1,\dots,x_{k-\ell})\in
A^{k-\ell}: \mathbf A\models
\phi(x_1,\dots,x_{k-\ell},a_1,\dots,a_\ell)\}.$
In this paper, by formula, we mean a first-order formula in
the language of $\Sigma$-structures, where $\Sigma$ is usually understood from the context.
Let $\mathbf A=(A,R_1({\mathbf A}),\dots,R_s({\mathbf A}))$ be a
$\Sigma$-structure and let $X\subseteq A$.
The \emph{substructure} of $\mathbf A$ \emph{induced by} $X$ is the
$\Sigma$-structure
$\mathbf A[X]=(X,R_1({\mathbf A})\cap X^{r_1},\dots,R_k({\mathbf
	A})\cap X^{r_s})$.

\emph{Graphs} are structures with a single binary relation $E$
encoding adjacency; this relation is irreflexive and symmetric.
Graphs of particular interest in this paper are rooted trees.  For a
rooted tree $Y$, we denote by $I(Y)$ the set of internal nodes of $Y$,
by $L(Y)$, the set of leaves of~$Y$, by $V(Y)=I(Y)\cup L(Y)$ the set of vertices of $Y$,
by $r(Y)$, the root of $Y$, and
by $\preceq_Y$, the partial order on $V(Y)$ defined by $u\preceq_Y v$
if the unique path in $Y$ linking $r(Y)$ and $v$ contains~$u$ (i.e.,
if $u=v$ or~$u$ is an \emph{ancestor} of~$v$ in~$Y$). For a non-root
vertex $v$, we further denote by $\pi_Y(v)$ the \emph{parent} of $v$,
which is the unique neighbor of $v$ smaller than $v$ with respect to
$\preceq_Y$. (We further define $\pi_Y(r(Y))=r(Y)$, so that $\pi_Y$ is defined on all the vertices of $Y$, the root being the only fixed point.)
A rooted \emph{binary} tree is a rooted tree such that
every internal node has exactly two children.

Let $Y$ be a rooted tree and let $A$ be a subset of vertices of $Y$ closed by pairwise least common ancestor (that is: the least common ancestor in $Y$ of any two vertices in $A$ also belongs to $A$). The \emph{subtree of}  $Y$ \emph{induced by} $A$ is the rooted tree $Y'$, whose associated tree order $\preceq_{Y'}$ is the restriction to $A$ of the  tree order $\preceq_Y$ associated to $Y$.  In particular, $A$ is the vertex set of $Y'$.

\emph{Partial orders} are structures with a single antisymmetric and
transitive binary relation~$\prec$.  Particular partial orders will be
of interest here.  \emph{Linear orders} (also called \emph{total orders}) are partial orders such that
$\forall x\,\forall y\ \bigl((x\prec y) \vee (y\prec x) \vee (y=x)\bigr)$.  \emph{Tree orders}
are partial orders that satisfy the following axioms:
$\forall x\,\forall y\,\forall z\ \big((x\prec z\wedge y\prec
z)\rightarrow((x\prec y)\vee (y\prec x)\vee (x=y))\big)$ and
$\exists r\,\forall x\ ((x=r)\vee (r\prec x))$. 
The minimum element of a tree order ($r$ in the previous equation) is its \emph{root}, and its maximal elements are its \emph{leaves}.
It will be convenient to use $\preceq, \succ, \succeq$
with their obvious meaning.  Let $(X,\prec)$ be a
tree order. The $\emph{infimum}$ $\inf(u,v)$ of two
elements $u,v\in X$ is the unique element $w\in X$ such that
$w\preceq u, w\preceq v$, and
$\forall z\ \bigl(((z\preceq u)\wedge (z\preceq v))\rightarrow (z\preceq
w)\bigr)$.  Note that $\inf(x,y)$ is first-order definable from $\prec$, 
hence can be
used as a term in our formulas.  
An element $x$ is \emph{covered} by an element $y$
if $x\prec y$ and there is no element $z$ with $x\prec z\prec y$. 
A \emph{binary tree order} is a tree order
such that every non-maximal element is covered by exactly two elements.

\emph{Ordered graphs} are structures with two binary relations, $E$
and $<$, where $E$ defines a graph and $<$ defines a linear order.  We
denote ordered graphs as $G^<=(V,E,<)$.

A {\em permutation} is represented as 
a structure $\sigma=(V,<_1,<_2)$, where $V$ is
a finite set and where~$<_1$ and~$<_2$ are two linear orders on this
set (see e.g.\ \cite{cameron2002homogeneous,albert2020two}).  Two
permutations $\sigma=(V,<_1,<_2)$ and $\sigma'=(V',<_1',<_2')$ are
{\em isomorphic} if there is a bijection between~$V$ and~$V'$
preserving both linear orders.  Let $X\subseteq V$. The {\em
	sub-permutation} of $\sigma$ induced by~$X$ is the permutation
on~$X$ defined by the two linear orders of $\sigma$ restricted to~$X$.
The isomorphism types of the sub-permutations of a permutation
$\sigma$ are the {\em patterns} of $\sigma$.  A class $\mathscr P$ of
(isomorphism types~of) permutations is {\em hereditary} (or
\emph{closed}) if it is closed under taking sub-permutations.  A
\emph{permutation class} is a hereditary class of permutations. A
permutation class is \emph{proper} if it is not the class of all
permutations. Note that the terms ``class of permutations'' and ``permutation class'' are not equivalent, the second referring to a hereditary class of permutations, as it is customary (see e.g.~\cite{bona2012combinatorics}).

\subsection{Interpretations}

Let $\Sigma,\Sigma'$ be signatures.  A \emph{simple interpretation} (or, simply, an \emph{interpretation}, since we will only consider these in this article)
$\mathsf I$ of $\Sigma'$-struc\-tures in $\Sigma$-struc\-tures is
defined by a $\Sigma$-formula $\rho_0(x)$, and a $\Sigma$-formula
$\rho_{R'}({x}_1,\dots,{x}_k)$ for each $k$-ary relation symbol
$R'\in\Sigma'$.
Let~$\mathsf I$ be an interpretation of \mbox{$\Sigma'$-structures} in
$\Sigma$-structures, where $\Sigma'=\{R_1',\dots,R_s'\}$.  For each
$\Sigma$-structure~$\mathbf A$ we denote by
$\mathsf I(\mathbf A)=(\rho_0(\mathbf A),\rho_{R_1'}(\mathbf
A),\dots,\rho_{R_s'}(\mathbf A))$ the \mbox{$\Sigma'$-structure} interpreted
by $\mathsf I$ in $\mathbf A$. Similarly, for a class $\mathscr C$ of $\Sigma$-structures, we
denote by~$\mathsf I(\mathscr C)$ the set
$\{\mathsf I(\mathbf{A}) : \mathbf{A}\in \mathscr C\}$.

We denote by $\sh_{\Sigma^+\rightarrow\Sigma}$ (or simply $\sh$ 
when~$\Sigma$ and~$\Sigma^+$ are clear from context) the interpretation
that ``forgets'' the relations in $\Sigma^+\setminus\Sigma$ while
preserving all the other relations and the domain. For a
$\Sigma^+$-structure $\mathbf B$, the $\Sigma$-structure
$\sh(\mathbf B)$ is called the $\Sigma$-\emph{reduct} (or simply
\emph{reduct} if $\Sigma$ is clear from the context) of $\mathbf B$.
A class $\mathscr C$ is a \emph{reduct} of a class $\mathscr D$ if
$\mathscr C=\sh(\mathscr D)$. Conversely, a class $\mathscr D$ is an
\emph{expansion} of $\mathscr C$ if $\mathscr C$ is a reduct of
$\mathscr D$.

Another important interpretation is $\mathsf{Gaifman}_\Sigma$ (or
simply $\mathsf{Gaifman}$ when $\Sigma$ is clear from context), which
maps a \mbox{$\Sigma$-structure}~$\mathbf A$ to its \emph{Gaifman graph},
whose vertex set is $A$ and whose edge set is the set of all pairs of distinct
vertices included in a tuple of some relation.

Note that an interpretation of $\Sigma_2$-structures in
$\Sigma_1$-structure naturally defines an interpretation of
$\Sigma_2^+$-structures in $\Sigma_1^+$-struc\-tures if
$\Sigma_2^+\setminus\Sigma_2=\Sigma_1^+\setminus\Sigma_1$ by leaving
the relations in $\Sigma_1^+\setminus\Sigma_1$ unchanged (that is, by
considering $\rho_R(x_1,\dots,x_k)=R(x_1,\dots,x_k)$ for these
relations).

\subsection{Transductions}\label{sec:transductions}

Let $\Sigma,\Sigma'$ be signatures. A \emph{simple transduction} $\mathsf T$
from $\Sigma$-struc\-tures to \mbox{$\Sigma'$-structures} is
defined by a simple interpretation~$\mathsf I_{\mathsf T}$ of
$\Sigma'$-structures in $\Sigma^+$-structures, where~$\Sigma^+$ is a
signature obtained from $\Sigma$ by adding finitely many marks. For a
\mbox{$\Sigma$-structure}~$\mathbf A$, we denote by
$\mathsf T(\mathbf A)$ the set of all
$\mathsf I_{\mathsf T}(\mathbf B)$ where $\mathbf B$ is a
$\Sigma^+$-structure with reduct $\mathbf A$:
$ \mathsf T(\mathbf A)=\{\mathsf I_{\mathsf T}(\mathbf B): \sh(\mathbf
B)=\mathbf A\}$.
Let $k\in\mathbb N$. The \emph{$k$-blowing} of a 
$\Sigma$-structure~$\mathbf A$ is the $\Sigma'$-structure 
$\mathbf B=\mathbf A\bullet k$,
where $\Sigma'$ is the signature obtained from $\Sigma$ by adding a
new binary relation~$\sim$ encoding an equivalence relation. The
domain of $\mathbf A\bullet k$ is $B=A\times[k]$, and, denoting by~$p_1$ and $p_2$
the projections $A\times[k]\rightarrow A$ and $A\times [k]\rightarrow [k]$ we have, for all $x,y\in B$,
$\mathbf B\models x\sim y$ if $p_1(x)=p_1(y)$, and (for $R\in\Sigma$)
$\mathbf B\models R(x_1,\dots,x_s)$ if 
$\mathbf A\models R(p_1(x_1),\dots,p_1(x_s))$ and $p_2(x_1)=\ldots= p_2(x_s)$.  A \emph{copying
	transduction} is the composition of a $k$-blowing and a simple
transduction; the integer $k$ is the \emph{blowing factor} of the copying transduction $\mathsf T$ and is denoted by $\bl(\mathsf T)$.
It is easily checked that the composition of two copying
transductions is again a copying transduction. In the following by the term transduction we mean a copying transduction. Note that for every transduction $\mathsf T$ from $\Sigma$-structure to $\Sigma'$, for every $\Sigma$-structure $\mathbf A$ and for every $\Sigma'$-structure $\mathbf B\in\mathsf T(\mathbf A)$ we have $|B|\leq \bl(\mathsf T)\,|A|$. 

Let $\mathsf T,\mathsf T'$ be transductions from $\Sigma$-structures
to $\Sigma'$-struc\-tures, and let $\mathscr C$ be a class of
$\Sigma$-structures.
The transduction $\mathsf T'$ \emph{subsumes} the transduction
$\mathsf T$ \emph{on} $\mathscr C$ if
$\mathsf T'(\mathbf A)\supseteq\mathsf T(\mathbf A)$ for all
$\mathbf A\in\mathscr C$.  If $\mathscr C$ is a class of
$\Sigma$-structures, we define
$\mathsf T(\mathscr C)=\bigcup_{\mathbf A\in\mathscr C}\mathsf
T(\mathbf A)$.  We say that a class~$\mathscr D$ of
$\Sigma'$-structures is a \emph{$\mathsf T$-transduction} of
$\mathscr C$ if $\mathscr D\subseteq\mathsf T(\mathscr C)$ and, more
generally, the class~$\mathscr D$ is a \emph{transduction} of the
class $\mathscr C$, and we write
$\xy\xymatrix{\mathscr C\ar@{->>}[r]&\mathscr D}\endxy$, if there
exists a transduction $\mathsf T$ such that $\mathscr D$ is a
$\mathsf T$-transduction of~$\mathscr C$.  The negation of
$\xy\xymatrix{\mathscr C\ar@{->>}[r]&\mathscr D}\endxy$ is denoted by
$\xy\xymatrix{\mathscr C\ar@{->>}|{/}[r]& \mathscr D}\endxy$.
Note that we require only 
the inclusion of $\mathscr D$ in $\mathsf T(\mathscr C)$, and not the equality. 
The class
$\mathscr D$ is a \emph{$c$-bounded} $\mathsf T$-transduction of the
class $\mathscr C$ if, for every $\mathbf B\in\mathscr D$, there exists
$\mathbf A\in\mathscr C$ with $\mathbf B\in\mathsf T(\mathbf A)$ and
$|A|\leq c\,|B|$. 
Two classes $\mathscr C$ and $\mathscr D$ are
\emph{transduction equivalent}
if each is a transduction of the other.
	A \emph{transduction pairing} of two classes~$\mathscr C$
	and~$\mathscr D$ is a pair $(\mathsf D,\mathsf C)$ of (copying) transductions,
	such that 
	$
	\forall\mathbf A\in\mathscr C\,\exists \mathbf B\in\mathsf D(\mathbf A)\cap\mathscr D\,:\, \mathbf A\in\mathsf C(\mathbf B)$ and  
	\mbox{$\forall\mathbf B\in\mathscr D$} $\exists \mathbf A\in\mathsf C(\mathbf B)\cap\mathscr C\,:\, \mathbf B\in\mathsf D(\mathbf A).
	$
	
	We denote by
	$\xy\xymatrix{\mathscr
		C\ar@{->>}@/^/[r]\ar@{}[r]|{\rightleftharpoons}&\mathscr
		D\ar@{->>}@/^/[l]}\endxy$ the existence of a transduction pairing of
	$\mathscr C$ and $\mathscr D$.   Note that if $(\mathsf D,\mathsf C)$ is a transduction pairing, then $\mathscr D$ is a $\bl(\mathsf C)$-bounded $\mathsf D$-transduction of $\mathscr C$ and~$\mathscr C$ is a $\bl(\mathsf D)$-bounded $\mathsf C$-transduction of $\mathscr D$. The following easy lemma will be useful.

	\begin{lem}
		\label{lem:pairing}
		Assume $\mathscr D$ is a $\mathsf D$-transduction of\, $\mathscr C$,
		$\mathscr C$ is a $\mathsf C$-transduction of\, $\mathscr D$,
		 and 
		for every $\mathbf A\in\mathscr C$ and every $\mathbf B\in\mathsf D(\mathbf A)\cap \mathscr D$ we have $\mathbf A\in\mathsf C(\mathbf B)$. Then $(\mathsf D,\mathsf C)$ is a transduction pairing of\, $\mathscr C$ and $\mathscr D$. 
	\end{lem}
	
	\begin{proof}
		Let $\mathbf B\in\mathscr D$. As $\mathscr D$ is a $\mathsf D$-transduction of $\mathscr C$ there exists $\mathbf A\in\mathscr C$ with $\mathbf B\in\mathsf D(\mathbf A)$.	Then $\mathbf A\in\mathsf C(\mathbf B)\cap \mathscr C$. 
	\end{proof}

	Note that a transduction $\mathsf T$ from $\Sigma_1$-structures to
	$\Sigma_2$-struc\-tures naturally defines a transduction $\widehat{\mathsf T}$ from
	$\Sigma_1^+$-structures to $\Sigma_2^+$-structures if
	$\Sigma_2^+\setminus\Sigma_2=\Sigma_1^+\setminus\Sigma_1$ by leaving
	the relations in $\Sigma_1^+\setminus\Sigma_1$ unchanged.
	The transduction $\widehat{\mathsf T}$ is called the \emph{natural generalization} of $\mathsf T$ to
	$\Sigma_1^+$-structures.
	\subsection{Twin-width}
	\label{sec:tww} 
		
	In order to define twin-width, we first need to introduce some
	preliminary notions, which generalize the notion of trigraphs ({i.e.}, graphs with some red edges)
	introduced in~\cite{twin-width1}.
	Let $\Sigma$ be a binary relational
	signature.  The signature $\Sigma^*$ is obtained by adding, for each
	binary relation symbol $R$ a new binary relation symbol~$R^\ast$. The symbol~$R^\ast$ will always be interpreted as 
	a symmetric relation 
	and plays for~$R$ the role of red edges
	in~\cite{twin-width1}.

	Let $\mathbf A$ be a $\Sigma^\ast$-structure, and let $u$ and $v$ be
	vertices of $\mathbf A$.  The vertices $u, v$ are \emph{$R$-clones}
	for a vertex $w$ and a relation $R\in\Sigma$ if we have
	$\mathbf A\models \big(R(u,w)\leftrightarrow R(v,w)\big)\wedge
	\big(R(w,u)\leftrightarrow R(w,v)\big)$ and no pair in $R^\ast$
	contains both $w$ and either $u$ or $v$.  The $\Sigma^*$-structure
	$\mathbf A'$ obtained by \emph{contracting} $u$ and $v$ into a new
	vertex $z$ is defined as follows:
	\begin{itemize}
		\item $A'=A\setminus\{u,v\}\cup\{z\}$;
		\item
		$R({\mathbf A'})\cap \bigl((A'\setminus\{z\})\times (A'\setminus\{z\})\bigr)=
		R({\mathbf A})\cap \bigl((A\setminus\{u,v\})\times (A\setminus\{u,v\})\bigr)$
		for all $R\in\Sigma^\ast$;
		\item for every vertex $w\in A'\setminus\{z\}$ and every $R\in\Sigma$
		such that $u$ and $v$ are $R$-clones for $w$, we let
		$\mathbf A'\models R(w,z)$ if $\mathbf A\models R(w,u)$, and
		$\mathbf A'\models R(z,w)$ if $\mathbf A\models R(u,w)$. (Note that this 
		does not change if we use $v$ instead of $u$);
		\item otherwise, for every vertex $w\in A'\setminus\{z\}$ and every
		$R\in\Sigma$ such that $u$ and $v$ are not $R$-clones for $w$ we
		let $\mathbf A'\models R^\ast(w,z)\wedge R^\ast(z,w)$.
	\end{itemize}
	
	A \emph{$d$-sequence} for a $\Sigma$-structure $\mathbf A$ is a
	sequence $\mathbf A_n,\dots,\mathbf A_1$ of $\Sigma^*$-structures
	such that: $\mathbf A_n$ is isomorphic to $\mathbf A$ (when considered as a $\Sigma^*$-structure with empty $R^*$); $\mathbf A_1$
	is the $\Sigma^*$-structure with a single element; for every
	$1\leq i<n$, $\mathbf A_i$ is obtained from $\mathbf A_{i+1}$ by
	performing a single contraction; for every $1\leq i<n$ and every
	$v\in A_i$, the sum of the degrees in relations
	$R^\ast\in\Sigma^\ast\setminus\Sigma$ of~$v$ in $\mathbf A_i$ is less
	or equal to $d$ (the degree of $v$ in relation $R^\ast$ is defined as
	the degree of $v$ in the undirected graph $(A,R^\ast(\mathbf A))$).
	When $d$ is not specified, we shall speak of a \emph{contraction sequence}; see \Cref{fig:tww} for an illustration.
The minimum $d$ such that there exists a $d$-sequence for a
$\Sigma$-structure~$\mathbf A$ is the \emph{twin-width}
$\tww(\mathbf A)$ of $\mathbf A$.   %
This definition for binary relational structures differs from the one given in \cite{twin-width1} (where red edges are not counted with multiplicity), but will be more convenient in our setting.
However, the definitions differ by at most a constant factor (linear in $|\Sigma|$), thus the derived notion of class with bounded twin-width coincides. 

A~crucial property of twin-width is the following result.
\begin{thmC}[{\cite[Theorem 8.1]{twin-width1}}]
	\label{thm:trans}
	Let $\mathscr C, \mathscr D$ be classes of binary structures.  If\,
	$\mathscr C$ has bounded twin-width and~$\mathscr D$ is a
	transduction of\, $\mathscr C$, then $\mathscr D$ has bounded
	twin-width.
\end{thmC}

\section{Classes with bounded star chromatic number}
One of the key ingredients of the proof will rely on a transduction pairing between a class~$\mathscr C$ and the class of Gaifman graphs of the structures in $\mathscr C$. Though such a pairing does not exist for general classes of structures (see the discussion below), we prove in this section that this is the case if the Gaifman graphs have bounded star chromatic number.

Recall that a \emph{star coloring} of a graph $G$ is a proper coloring of $G$ such that 
any two color
classes induce a star forest (i.e., a disjoint union of stars); 
the \emph{star chromatic 
	number} $\chi_{\rm st}(G)$ of  $G$ is the minimum
number of colors in a star coloring of~$G$.
Note that a star coloring of a graph with $c$ colors defines a partition of the edge set into $\binom{c}{2}$ star forests.
Although we are interested only in binary relational structures in this paper,
the next lemma holds (and is proved) for general relational
signatures.

\begin{figure*}
	\begin{center}
		\includegraphics[width=0.95\textwidth]{tww.pdf}
		\caption{A contraction sequence, a so-called block representation of the contractions, and a twin-model.}
		\label{fig:tww}
	\end{center}	
\end{figure*}

\begin{lem}
	\label{lem:Gaifman}
	Let $\Sigma$ be a relational signature, let $\mathscr C$ be a class of
	$\Sigma$-structures, and let~$c$ be an integer. There exists a simple 
	transduction $\mathsf{Unfold}_{\Sigma,c}$ from graphs to $\Sigma$-structures
	such that if the Gaifman graphs of the structures in $\mathscr C$ have
	star chromatic number at most $c$, then 
	$(\mathsf{Gaifman}_\Sigma,\mathsf{Unfold}_{\Sigma,c})$ is a transduction pairing of $(\mathscr C,\mathsf{Gaifman}_\Sigma(\mathscr C))$.
\end{lem}
\begin{proof}
	Let
	$c=\sup\{\chi_{\rm st}(G): G\in\mathsf{Gaifman}_\Sigma(\mathscr C)\} <
	\infty$.  Let $\mathbf A\in\mathscr C$, let
	$G=\mathsf{Gaifman}_\Sigma(\mathbf A)$, and let
	$\gamma:V(G)\rightarrow [c]$ be a star coloring of $G$. In $G$, any
	two color classes induce a star forest, which we orient away from
	their centers. This way we get an orientation $\vec G$ of $G$ such
	that for every vertex $v$ and every in-neighbor $u$ of $v$, the
	vertex $u$ is the only neighbor of $v$ with color~$\gamma(u)$.
	Let $R\in\Sigma$ be a relation of arity $k$.  For each
	$(u_1,\dots,u_k)\in R(\mathbf A)$, $u_1,\ldots, u_k$ induce a 
	tournament in~$\vec G$. Every tournament
	has at least one directed Hamiltonian path~\cite{redei1934kombinatorischer}.
	We fix one such Hamiltonian path and let
	$p(u_1,\dots,u_k)$ be the index of the last vertex in the
	path. 
	Let $a=p(u_1,\dots,u_k)$, let
	$(c_1,\dots,c_k)=(\gamma(u_1),\dots,\gamma(u_k))$. Then there exists
	in $G$ exactly one clique of size $k$ containing $u_a$ with vertices
	colored $c_1,\dots,c_k$, as a consequence of the following claim.
	(In the claim, $\gamma(K)=\{\gamma(v): v\in K\}$.)
	\begin{claim}
		Assume $K_1,K_2$ are two  $k$-cliques with $\gamma(K_1)=\gamma(K_2)$. 
		If there exists in $\vec G$ a directed Hamiltonian path $\vec P$ of $\vec{G}[K_1]$ ending at a vertex $v\in K_1\cap K_2$, then $K_1=K_2$.
	\end{claim}
\begin{claimproof}
	We prove the statement by induction on $k$. If $k=1$ the statement is obviously true as $K_1=K_2=\{v\}$. Assume that the statement holds for some integer $k\geq 1$, let $K_1,K_2$ be $(k+1)$-cliques with $\gamma(K_1)=\gamma(K_2)$ and assume there exists a directed Hamiltonian path $\vec P$ of $\vec{G}[K_1]$ ending at a vertex $v\in K_1\cap K_2$. Let $u$ be the penultimate vertex of $\vec P$. As $\gamma(K_1)=\gamma(K_2)$, there exists a vertex $u'\in K_2$ with $\gamma(u')=\gamma(u)$. Note that $u'\neq v$ as $\gamma(u)\neq \gamma(v)$. 	As $u$ is an in-neighbor of $v$, it is the only neighbor of $v$ in $G$ with color~$\gamma(u)$. As $K_2$ induces a clique, $u'$ is a neighbor of $v$. Hence, $u'=u$.
	Let $K_1'=K_1\setminus\{v\}$ and $K_2'=K_2\setminus\{v\}$. Then, $K_1'$ and $K_2'$ are $k$-cliques with $\gamma(K_1')=\gamma(K_2')=\gamma(K_1)\setminus\{\gamma(v)\}$ and $\vec P-\{v\}$ is a directed Hamiltonian path of $\vec G[K_1']$ ending at $u\in K_1'\cap K_2'$. By the induction hypothesis we have
	$K_1'=K_2'$, hence $K_1=K_2$.
\end{claimproof}

	For each relation $R\in\Sigma$
	with arity $k$ and each $(u_1,\dots,u_k)\in R(\mathbf A)$ we put at
	$v=u_{p(u_1,\dots,u_k)}$ a mark
	$M^R_{\gamma(u_1),\dots,\gamma(u_k)}$. 
	Then, in the graph~$G$, the vertex $v$ belongs to exactly one clique of size~$k$ with vertices colored
	$\gamma(u_1),\dots,\gamma(u_k)$, which allows recovering the tuple $(u_1,\dots,u_k)$, as all the colors are distinct.
	We further put at each vertex
	$v$ a mark~$C_{\gamma(v)}$. Then the structure $\mathbf A$ is
	reconstructed by the transduction $\mathsf{Unfold}_{\Sigma,c}$ defined by the formulas
	\begin{align*}
	\rho_R(x_1,\dots,x_k):=
	\bigvee_{c_1,\dots,c_k}\Bigl(\bigwedge_{1\leq j\leq
		k}C_{c_i}(x_i)\wedge\bigwedge_{1\leq i<j\leq k}E(x_i,x_j)\hspace{1mm} \wedge 
	\bigvee_{1\leq i\leq k}M^R_{c_1,\dots,c_k}(x_i)\Bigr). \tag*{\qedhere}
	\end{align*}
\end{proof}
Note that the condition of \Cref{lem:Gaifman} is almost tight:
if a class $\mathscr C$ of undirected graphs contains graphs with arbitrarily large star chromatic number and girth, then the class $\vec{\mathscr C}$ of all orientations of the graphs in $\mathscr C$ is not a transduction of~$\mathscr C$~\cite{SODA_msrw}.

\Cref{lem:Gaifman} will be particularly significant in conjunction with the
following results. 
Recall that a graph $G$ is \mbox{\emph{$d$-degenerate}} if every induced subgraph of $G$ contains a vertex of degree at most~$d$, and that a class $\mathscr C$ is \emph{degenerate} if all the graphs in $\mathscr C$ are $d$-degenerate for some~$d$. A class $\mathscr C$ of graphs has \emph{bounded expansion} if, for every integer $k$, the class  of all graphs $H$ whose $k$-subdivision is a subgraph of some graph in $\mathscr C$ is degenerate \cite{Sparsity}.

\begin{thmC}[\cite{twin-width2}]
\label{thm:sptww}
Every degenerate class of graphs with bounded twin-width has bounded expansion.
\end{thmC}

\begin{thmC}[\cite{POMNI}]
\label{thm:BEchist}
Every class of graphs with bounded expansion has bounded star chromatic number.
\end{thmC}

\section{Twin-models}
In this section, we formalize the notions of twin-models and ranked
twin-models, which are reminiscent of the ``ordered union trees'' and
``interval biclique partitions'' adopted in~\cite{twin-width3}. This
structure will allow encoding a contraction sequence and to give an
alternative definition of twin-width.
As mentioned in the introduction, we fix a class $\mathscr C_0$ of binary relational structures with bounded twin-width.

\subsection{Twin-models, ranking, layers, and width}

\label{sec:twm}

\begin{defi}[twin-model]
	Let $\Sigma=(R_1,\dots,R_k)$ be a binary relational signature.  A
	{\em $\Sigma$-twin-model} (or simply a \emph{twin-model} when
	$\Sigma$ is clear from the context) is a tuple
	$(Y,Z_{R_1},\dots,Z_{R_k})$ where $Y$ is a rooted binary tree and
	each~$Z_{R_i}$ is a binary relation satisfying the following transversality, 
	minimality,  and consistency conditions:
	\begin{itemize}
		\item {\rm (transversality)} if $(u,v)\in Z_{R_i}$, then $u$ and $v$ are not comparable in the tree order $\preceq_Y$;
		\item {\rm (minimality)} if $(u,v)\in Z_{R_i}$, then there exists no
		$(u',v')\neq (u,v)$ with $u'\preceq_Y u$, $v'\preceq_Y v$ and
		$(u',v')\in Z_{R_i}$;
		\item {\rm (consistency)} if a traversal of a cycle  $\gamma$ in $Y\cup \bigcup_i Z_{R_i}$ respects the \emph{natural orientation} of 
		the $Y$-edges (that is: the orientation of $Y$ away from the root), then $\gamma$ contains two consecutive edges
		in $\bigcup_i Z_{R_i}$.
	\end{itemize}
	
	A twin-model $(Y,Z_{R_1},\dots,Z_{R_k})$ {\em defines} the
	$\Sigma$-structure $\mathbf A$ (or $(Y,Z_{R_1},\dots,Z_{R_k})$ is a
	{\em twin-model of}~$\mathbf A$) if $A=L(Y)$ and, for each
	$R_i\in\Sigma$, $R_i({\mathbf A})$ is the set of all pairs $(u, v)$
	such that there exists $u'\preceq_Y u$ and $v'\preceq_Y v$ with
	$(u',v')\in Z_{R_i}$.
\end{defi}

\begin{defi}[ranking, boundaries, layers, and width]
	Let $(Y,Z_{R_1},\dots,Z_{R_k})$ be a twin-model of a
	$\Sigma$-structure $\mathbf A$ with $|A|=n$.  A {\em ranking} $\tau$
	of the twin-model $(Y,Z_{R_1},\dots,Z_{R_k})$ is a mapping from
	$V(Y)$ to $[n]$ that satisfies the following labeling, monotonicity,
	and synchronicity conditions:
	\begin{itemize}
		\item {\rm (labeling)} the function $\tau$ restricted to $I(Y)$ is a
		bijection with $[n-1]$, and is equal to $n$ on~$L(Y)$;
		\item {\rm (monotonicity)} If $u\prec_Y v$, then $\tau(u)<\tau(v)$;
		\item {\rm (synchronicity)} If $(u,v)\in Z_{R_i}$, then 
		$\max(\tau(\pi_Y(u)),\tau(\pi_Y(v)))<\min(\tau(u),\tau(v))$.
	\end{itemize}
	
	A {\em ranked twin-model} is a tuple
	$\mathfrak T=(Y,Z_{R_1},\dots,Z_{R_k},\tau)$, where
	$(Y,Z_{R_1},\dots,Z_{R_k})$ is a twin-model, and $\tau$ is a ranking
	of $(Y,Z_{R_1},\dots,Z_{R_k})$ (See \Cref{fig:tm}).
	
	\begin{figure}[h!t]
		\hfill
		\begin{minipage}{.15\columnwidth}
			\includegraphics[width=\textwidth]{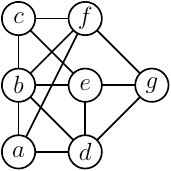}
		\end{minipage}\hfill
		\begin{minipage}{.75\columnwidth}
			\includegraphics[width=\textwidth]{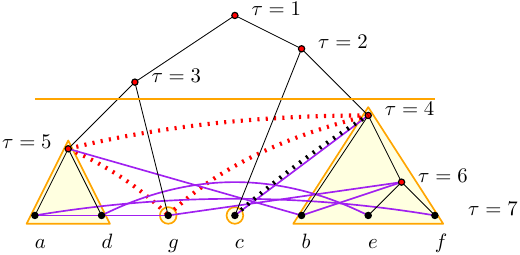}
		\end{minipage}
		\hfill{}
		\caption{A graph $G$ and a ranked twin-model of $G$. The boundary $\partial_4Y$ is the set $\{5,g,c,4\}$ (internal vertices labeled by $\tau$), which can be represented as the set of the yellow zones. The relations of $\mathbf L_4$ are depicted as dotted heavy lines (black for $R$, red for~$R^\ast$). The width of this twin-model is $2$.}
		\label{fig:tm}
	\end{figure}

	\smallskip For $1<t\leq n$, the {\em boundary} $\partial_t Y$ is the
	set
	$ \partial_t Y=\{u\in V(Y)\mid \tau(u)\geq t \wedge \tau(\pi_Y(u))<t\}
	$ and the {\em layer} $\mathbf L_t$ is the $\Sigma^\ast$-structure
	with vertex set $\partial_t Y$ and relations
	\begin{align*}
		R_i(\mathbf L_t) & =\{(u,v)\in \partial_t Y\times \partial_t Y\mid \exists u'\preceq_Y u,\ \exists v'\preceq_Y v,\ (u',v')\in Z_{R_i}\}\\ 	
		R_i^\ast(\mathbf L_t) & =\{(u,v)\in \partial_t Y\times \partial_t Y\mid \exists u'\succeq_Y u,\  \exists v'\succeq_Y v,\ (u',v')\neq(u,v)\\
		&\hspace{7cm} \text{ and }\{(u',v'),(v',u')\}\cap Z_{R_i}\neq\emptyset\}.
	\end{align*}
	
	For $t=1$ we define the boundary $\partial_1 Y=\{r(Y)\}$ and the layer
	$\mathbf L_1$ as the $\Sigma^\ast$-structure with unique vertex
	$r(Y)$.
	
	The {\em width} of the ranked twin-model
	$\mathfrak T=(Y,Z_{R_1},\dots,Z_{R_k},\tau)$ is defined as
	\[
	\text{\rm width}(\mathfrak T)=\max_{t\in [n]}\max_{v\in
		L_t}\sum_{R_i\in\Sigma}|R_i^\ast(\mathbf L_t,v)|,
	\]
	where $R_i^\ast(\mathbf L_t,v)$ denotes the set $\{u: \mathbf L_t\models R_i^\ast(u,v)\}$. Hence,  $|R_i^\ast(\mathbf L_t,v)|$ is the degree of $v$ in the symmetric relation $R_i^\ast$ in $\mathbf L_t$.
\end{defi}

At first sight, the consistency condition of a twin-model (of $\mathbf A$) may seem contrived.
One may for instance wonder if the minimality and consistency conditions are not simply equivalent to the property that every $(u,v) \in R_i(\mathbf A)$ is \emph{realized} by a unique unordered pair~$u',v'$ with $(u',v') \in Z_{R_i}$, $u' \preceq_Y u$, and $v' \preceq_Y v$.
In case the structure $\mathbf A$ encodes a simple undirected graph $G$ (with signature $\Sigma = (E)$), we would simply impose that the edges of $Z_E$ partition the edges of $G$ into bicliques.

In~\Cref{fig:why-consistency} we give a small example that shows that this property is not strong enough to always yield a ranking.
This illustrates why the consistency condition is what we want (no more, no less) and also serves as a visual support for the notions of contraction sequence, twin-model, and ranking.
\begin{figure}[h!t]
	\centering
	\scalebox{.6}{
		\begin{tikzpicture}[scale=0.9]
			\foreach \i/\j/\l in {0/0/a, 0/1/b, 2/0/c, 2/1/d, 1/3.2/e, 1/2.2/f}{
				\node[draw,circle,minimum size=0.66cm] (\l) at (\i,\j) {$\l$} ;
			}
			\foreach \i/\j in {b/c,b/d,d/e,d/f,f/a,f/b}{
				\draw (\i) -- (\j) ;
			}
			\foreach \i/\j/\l in {a/b/ab, c/d/cd, e/f/ef, ab/cd/abcd, abcd/ef/all}{
				\node[draw,rounded corners,fit=(\i) (\j)] (\l) {} ;
			}
			\foreach \i/\j/\s in {1/0/2,0/2.5/1,0/0.5/5,2/0.5/3,1/2.7/4}{
				\node at (\i,\j) {\tiny{\s}} ;
			}
			
			\begin{scope}[xshift=2.4cm]
				\foreach \i/\l in {1/a,2/b,3/c,4/d,5/e,6/f}{
					\node[draw,circle,inner sep=0.04cm,minimum size=0.5cm] (v\l) at (\i,0) {$\l$} ;
				}
				\foreach \i/\j/\l in {1.5/1.1/5,3.5/1.1/3,5.5/1.1/4, 2.5/2.2/2, 3.5/3.3/1}{
					\node[draw,circle,inner sep=0.04cm,minimum size=0.5cm] (v\l) at (\i,\j) {$\l$} ;
				}
				\foreach \i/\j in {a/5,b/5, c/3,d/3, e/4,f/4, 5/2,3/2, 2/1, 4/1}{
					\draw[<-] (v\i) -- (v\j) ;
				}
				\foreach \i/\j/\b in {b/c/0,b/d/-40,d/4/0,f/5/-5.5}{
					\draw[line width=0.05cm,blue] (v\i) to [bend left=\b] (v\j) ;
				}
			\end{scope}
			
			\begin{scope}[xshift=8.3cm]
				\foreach \i/\l in {1/a,2/b,3/c,4/d,5/e,6/f}{
					\node[draw,circle,inner sep=0.04cm,minimum size=0.5cm] (v\l) at (\i,0) {$\l$} ;
				}
				\foreach \i/\j/\l/\p in {1.5/1.1/5/$\alpha$,3.5/1.1/4/$\beta$,5.5/1.1/3/$\gamma$, 2.5/2.2/2/, 3.5/3.3/1/}{
					\node[draw,circle,inner sep=0.04cm,minimum size=0.5cm] (v\l) at (\i,\j) {\p} ;
				}
				\foreach \i/\j in {a/5,b/5, c/4,d/4, e/3,f/3, 5/2,4/2, 2/1, 3/1}{
					\draw[<-] (v\i) -- (v\j) ;
				}
				\foreach \i/\j/\b in {b/4/0,d/3/0,f/5/-5.5}{
					\draw[line width=0.05cm,blue] (v\i) to [bend left=\b] (v\j) ;
				}
			\end{scope}
	\end{tikzpicture}}
	\caption{Left: A 6-vertex graph and a contraction sequence, where the tiny digit in each box indicates the index of contracted vertices when they appear. Center: A twin-model of the graph, where the edges of $Z_E$ are in bold blue, and a ranking (for the internal nodes) of this twin-model that actually matches the contraction sequence. Right: A \emph{flawed twin-model} where the edge set $E$ is indeed partitioned by the pairs of~$Z_E$. Here no ranking is possible: 
		Let $\alpha,\beta,\gamma$ be the parents of $b,d,f$. By synchronicity, and by considering the pairs $(b,\beta)$, $(d,\gamma)$, and $(f,\alpha)$, we get that the labeling $\tau$ should satisfy 
		$\tau(\alpha)<\tau(\beta)$, $\tau(\beta)<\tau(\gamma)$, and $\tau(\gamma)<\tau(\alpha)$, which cannot be realized.
		There is indeed a cycle $\alpha b\beta d\gamma f$ with all the tree arcs oriented the same way, and without two consecutive edges of $Z_E$.
		On the contrary, all such cycles in the central tree have two consecutive edges of $Z_E$, like $5bd4f$ has $(b,d),(d,4) \in Z_E$.
	}
	\label{fig:why-consistency}
\end{figure}

\subsection{From a contraction sequence to a twin-model}

In this section, we prove that every $d$-sequence of a
$\Sigma$-struc\-ture~$\mathbf A$ defines a ranked twin-model of
$\mathbf A$ with width at most~$d$ (See \Cref{fig:seq2tm}).

A $d$-sequence $\mathbf A_n,\dots,\mathbf A_{1}$ for a
$\Sigma$-structure $\mathbf A$ defines a rooted binary tree $Y$ with
vertex set $V(Y)=\bigcup_i A_i$ and set of leaves $L(Y)=A_n$ as
follows: for each $i\in[n-1]$ let $z_i$ be the vertex of $A_i$ and
$u_i,v_i$ be the vertices of $A_{i+1}$ such that~$z_i$ results from
the contraction of $u_i$ and $v_i$ in $\mathbf A_{i+1}$.
Then $I(Y)=\{z_i: i\in [n-1]\}$, $r(Y)=z_1$, and the children of
$z_i$ in~$Y$ are the vertices $u_i$ and $v_i$.

For each relation $R\in \Sigma$ we define a binary relation $Z_R$ on
$V(Y)$ as follows.  Let $z_i$ be the vertex of $A_i$ resulting from
the contraction of $u_i$ and $v_i$ in $A_{i+1}$.  If
$(u_i,v_i)\in R({\mathbf A_{i+1}})$, then $(u_i,v_i)\in Z_R$.  If
$u_i$ and $v_i$ are not $R$-clones for $w$, then the pairs involving
$w$ and $u_i$ or $v_i$ in~$R({\mathbf A_{i+1}})$ are copied in~$Z_R$.
Intuitively, $Z_R$ collects the $R$-relations when they just appear
(in the order $\mathbf A_1,\ldots, \mathbf A_n$). We further define
$Z=\bigcup_{R\in\Sigma}Z_R$ and the function
$\tau\colon V(Y)\rightarrow [n]$ by $\tau(v)=n$ if $v\in L(Y)$ and
$\tau(z_i)=i$. Note that for each $i\in [n]$ and non-root vertex $v$
of~$Y$, we have $v\in A_i$ if and only if
$\tau(\pi_Y(v))<i\leq \tau(v)$.

\begin{figure}[h!t]
	\includegraphics[width=.4\columnwidth]{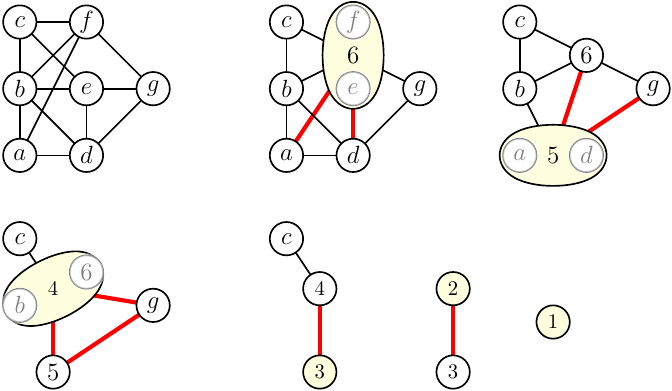}\hfill\includegraphics[width=.55\columnwidth]{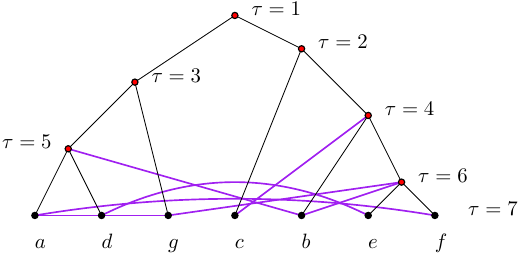}
	\caption{A contraction sequence and the derived ranked tree model.}
	\label{fig:seq2tm}
\end{figure}

\begin{lem}
	\label{lem:seq2tm}
	Every $d$-sequence $\mathbf A_n,\dots,\mathbf A_1$ defines a ranked twin-model with width at most~$d$.
\end{lem}

\begin{proof}
	\begin{claim}
		\label{cl:tau}
		The function $\tau$ satisfies the labeling, monotonicity, and
		synchronicity conditions.
		
	\end{claim}
	\begin{claimproof}
		The first two conditions are straightforward. Let \mbox{$(u,v)\in Z_R$}.
		Let $i\in[n-1]$ be such that $(u,v)$ appears in $\mathbf A_i$ for
		the first time. As $u,v\in A_i$ we have both
		$\tau(\pi_Y(u))<i\leq \tau(u)$ and $\tau(\pi_Y(v))<i\leq \tau(v)$,
		i.e., the synchronicity condition holds.
	\end{claimproof}
	
	\begin{claim}
		\label{cl:Z}
		The relations $Z_R$ ($R\in \Sigma$) satisfy the minimality and
		consistency conditions.
		
	\end{claim}
	\begin{claimproof}
		The minimality condition follows directly from the definition.  Let
		$\vec H$ be the oriented graph obtained from $Y$ by orienting all
		the edges from the root and adding, for each $R\in\Sigma$ and each pair $(u,v)\in Z_R$ the
		arcs $\pi_Y(u)v$ and $\pi_Y(v)u$ whenever they do not exist. It
		follows from the monotonicity and synchronicity conditions that
		$\vec H$ is acyclically oriented.  Indeed, any arc $(x,y)$ in
		$\vec H$ satisfies $\tau(x) < \tau(y)$.
		
		Assume towards a contradiction that in $Y\cup \bigcup_{R\in\Sigma}Z_R$ one can find a
		cycle $\gamma$ such that the orientation of the $Y$-edges is consistent with a traversal of $\gamma$ and $\gamma$ does not 
		contain two consecutive edges in~$\bigcup_{R\in\Sigma}Z_R$. By replacing in $\gamma$
		each group formed by an edge in~$\bigcup_{R\in\Sigma}Z_R$ and its preceding edge in $\gamma$ (which is in $Y$)
       by the corresponding arc in $\vec H$ we obtain a circuit in
		$\vec H$, contradicting its acyclicity. Hence,~the relations $Z_R$ satisfy the
		consistency condition.
	\end{claimproof}

		From the definition of the width of a ranked twin-model, it is then immediate  that	the ranked twin-model $(Y,Z_{R_1},\dots,Z_{R_k},\tau)$ derived from a $d$-sequence $\mathbf A_n,\dots,\mathbf A_1$ has width at most $d$. This ends the proof of the lemma.
\end{proof}

\subsection{Properties of twin-models}

In this section, we establish two properties of twin-models. The first one is the equality of the minimum width of a twin-model with the twin-width of a structure; the second one is that twin-models of structures with bounded twin-width have degenerate Gaifman graphs.

\begin{restatable}{lem}{Ltwwtm}
	\label{lem:twwtm}
	Every twin-model has a ranking, and 
	the twin-width of a $\Sigma$-structure $\mathbf A$ is the minimum
	width of a ranked twin-model of $\mathbf A$.
\end{restatable}
\begin{proof}
		We first prove the first part of the statement.
	
	\begin{claim}
		Every twin-model has a ranking.
	\end{claim}
	\begin{claimproof}
		Consider the oriented graph $\vec H$ obtained from orienting $Y$
		from the root and adding, for each  $R\in\Sigma$ and each pair $(u,v)\in Z_R$, 
		an arc
		$\pi(u)v$ and an arc $\pi(v)u$ (whenever they do not exist). Assume
		for contradiction that $\vec H$ contains a directed cycle. Replace
		each arc of the form $\pi(u)v$ of this directed cycle (with $(u,v)\in Z_R$)
		by the path $(\pi(u)u, uv)$ in the twin-model. This way we obtain a
		closed walk in $Y\cup \bigcup_{R\in\Sigma} Z_R$ traversing all edges of $Y$ away 
		from the root and no two consecutive edges are in $\bigcup_{R\in\Sigma}Z_R$. We show that we can also find a directed cycle  in $Y\cup \bigcup_{R\in\Sigma}Z_R$
		 with this 
		property, contradicting the
		consistency assumption. 
		Consider a shortest closed walk $W=(e_1,\dots,e_m)$ with the above property and assume this closed walk is not a directed cycle. Without loss of generality we can assume that $(e_1,\dots,e_k)$ forms a cycle~$\gamma$ (starting the closed 
		walk at another point if necessary). By minimality of the closed walk, the cycle $\gamma$ contains two consecutive edges in  $\bigcup_{R\in\Sigma}Z_R$. These edges are the edges~$e_k$ and~$e_1$ (as otherwise they would be consecutive in $W$ as well). It follows that~$e_{k+1}$ does not belong to  $\bigcup_{R\in\Sigma}Z_R$ (as it follows $e_k$ in the~$W$). The closed walk $W'=(e_{k+1},\dots,e_n)$ does not  have two consecutive edges in  $\bigcup_{R\in\Sigma}Z_R$ as all the consecutive pairs are consecutive in $W$, except the pair $(e_n,e_{k+1})$ (and we know $e_{k+1}\notin\bigcup_{R\in\Sigma}Z_R$). This contradicts the minimality of~$W$.
		Thus, $\vec H$ is acyclic and a topological
		ordering of $\vec H[I(Y)]$ extends to a labeling
		$\tau\colon V(H)\rightarrow[n]$ that is bijective between~$I(Y)$ and
		$[n-1]$, equal to $n$ on~$L(Y)$, and increasing with respect to
		every arc of $\vec H$. This directly implies both monotonicity
		and synchronicity.
		
		\vspace{-4mm}
	\end{claimproof}

	The
	following claim, which asserts that no $Z_{R_i}$ ``crosses'' the
	boundaries, will be quite helpful.
	\begin{claim}
		\label{cl:nocross}
		Let $t\in [n-1]$ and let $u,v\in\partial_t Y$. Then there exists no
		pair $(u',v')\in Z_{R_i}$ with $u'\prec_Y u$ and $v'\succ_Y v$.
	\end{claim}
	\begin{claimproof}
		Assume $(u',v')\in Z_{R_i}$ and $v'\succ_Y v$. By the synchronicity
		property we have $\tau(u')>\tau(\pi_Y(v'))\geq \tau(v)\geq t$,
		contradicting $\tau(u')\leq \tau(\pi_Y(u))<t$.
	\end{claimproof}
	
	For a $\Sigma^\ast$-structure $\mathbf A$ and $R\in\Sigma$ we define
	\[
	\overline{R}(\mathbf A)=\{(u,v)\in A^2:\{(u,v),(v,u)\}\cap
	(R(\mathbf A)\cup R^\ast(\mathbf A))\neq \emptyset\}.
	\]
	
	\begin{claim}
		\label{cl:l2s}
		Let $\mathbf L_1,\dots,\mathbf L_n$ be the layers of a ranked 
		twin-model~$\mathfrak T$ of a $\Sigma$-structure $\mathbf A$.  Then there
		exists a contraction sequence $\mathbf A_n,\dots,\mathbf A_1$ of
		$\mathbf A$ with $A_i=L_i$, and, for each $R\in\Sigma$, \linebreak
		\mbox{$\overline{R}(\mathbf A_i)=\overline{R}(\mathbf L_i)$},
		$R(\mathbf L_i)\subseteq R(\mathbf A_i)$ and
		$R^\ast(\mathbf L_i)\supseteq R^\ast(\mathbf A_i)$.
	\end{claim}
	\begin{claimproof}
		For $i\in [n-1]$, the $\Sigma$-structure $\mathbf A_{i}$ is obtained
		from~$\mathbf A_{i+1}$ by contracting the pair of vertices
		$u_i, v_i$ into $w_i$, where $w_i$ is the vertex of $Y$ with
		$\tau(w_i)=i$ and $u_i$ and $v_i$ are the two children of $w_i$ in
		$Y$. It is easily checked that $A_i=L_i$. Let $z$ be a vertex of
		$A_i$ different from~$w_i$. Then
		$(z,w_i)\in \overline{R}(\mathbf A_i)$ if and only if there exists a leaf
		$w'\succeq_Y w_i$ and a leaf $z'\succeq_Y z$ such that
		$\{(w',z'),(z',w')\}\cap R(\mathbf A)\neq\emptyset$. As
		$\mathfrak T$ is a twin-model of~$\mathbf A$ this is equivalent 
		to the fact that there
		exists $w''\preceq_Y w'$ and $z''\preceq_Y z'$ with
		$(w'',z'')\in Z_R$ or \mbox{$(z'',w'')\in Z_R$}.  As~$\preceq_Y$ is a
		tree order, $w_i$ and $w''$ are comparable, as well as $z$ and
		$z''$. From this and Claim~\ref{cl:nocross} it follows that 
		$(z,w_i)\in \overline{R}(\mathbf{A}_i)\iff
		\{(z,w_i),(w_i,z)\}\subseteq \overline{R}(\mathbf L_i)\iff 
		(z,w_i)\in \overline{R}(\mathbf{L}_i)$, thus
		$\overline{R}(\mathbf{A}_i)=\overline{R}(\mathbf{L}_i)$. 
		
		We now prove $R(\mathbf L_i)\subseteq R(\mathbf A_i)$ by reverse
		induction on $i$. For $i=n$ we have
		$R(\mathbf L_i)=R(\mathbf A_i)=R(\mathbf A)$. Let $i\in [n-1]$ and
		let $u_i,v_i,w_i$ be defined as above.  If
		$(w_i,z)\in R(\mathbf L_i)$, then there exists $w'\preceq_Y w_i$ and
		$z'\preceq_Y z$ with $(w',z')\in Z_R$ thus we have also
		$(u_i,z)\in R(\mathbf L_{i+1})$ and $(v_i,z)\in R(\mathbf L_{i+1})$.
		By induction, we deduce $(u_i,z)\in R(\mathbf A_{i+1})$ and
		$(v_i,z)\in R(\mathbf A_{i+1})$.  Similarly, if
		$(z,w_i)\in R(\mathbf L_i)$, then $(z,u_i)\in R(\mathbf A_{i+1})$
		and $(z,v_i)\in R(\mathbf A_{i+1})$.  Thus, $u_i$ and~$v_i$ are
		$R$-clones for $z$, hence, if $(w_i,z)\in R(\mathbf L_i)$, then
		$(w_i,z)\in R(\mathbf A_i)$ and if $(z,w_i)\in R(\mathbf L_i)$, then
		\mbox{$(z,w_i)\in R(\mathbf A_i)$}.  It follows that we have
		$R(\mathbf L_i)\subseteq R(\mathbf A_i)$. Thus, we have
		\begin{align*}
			R^\ast(\mathbf L_i)&=\overline R(\mathbf L_i)\setminus\bigl\{(u,v): \{(u,v),(v,u)\}\cap R(\mathbf L_i)=\emptyset\bigr\}\\
			&\supseteq \overline R(\mathbf A_i)\setminus\bigl\{(u,v): \{(u,v),(v,u)\}\cap R(\mathbf A_i)=\emptyset\bigr\} \\
			& =R^\ast(\mathbf A_i).\tag*{\qedhere}
		\end{align*}
		
	\end{claimproof}

	We are now able to complete the proof of the lemma.
	According to \Cref{lem:seq2tm}, every
	$d$-sequence for $\mathbf A$ defines a ranked twin-model with 
	width at most $d$. Conversely, every ranked twin-model for
	$\mathbf A$ with width $d'$ defines a sequence of layers
	$\mathbf L_t$ with 
	$\max_{v\in L_t} \sum_{R_i\in\Sigma}|R_i^\ast(\mathbf L_t,v)|$ $\leq
	d'$ and, by Claim~\ref{cl:l2s}, a $d'$-sequence for $\mathbf A$.
\end{proof}

\Cref{lem:twwtm} allows introducing the following terminology: the \emph{width} of a twin-model $(Y,Z_{R_1},\dots,Z_{R_k})$ is the minimum width of a ranking of $(Y,Z_{R_1},\dots,Z_{R_k})$.
A twin-model of a \mbox{$\Sigma$-structure} $\mathbf A$ is \emph{optimal} if it has the minimum possible width as a twin-model of  $\mathbf A$, which is the twin-width of $\mathbf A$.

\begin{defi}[The class $\mathscr T$]
	\label{def:clT}
	The class $\mathscr T$ is the class of all optimal twin-models of the $\Sigma$-structures in $\mathscr C_0$.
\end{defi}

\pagebreak
The following easy remark will be useful.
\begin{claim}
	\label{cl:indtm}
	Let $\mathfrak Y=(Y,Z_{R_1},\dots,Z_{R_k},\tau)$ be a ranked
	twin-model of a $\Sigma$-structure $\mathbf A$ (with domain $A$) and let $X\subseteq A$.
	Let $Y'$ be the subtree of $Y$ induced by all the vertices in~$X$ and
	their pairwise least common ancestors in $Y$, let $Z_{R_i}'$ be the subset of all pairs in
	$Z_{R_i}\cap (Y'\times Y')$, and let $\tau'$ be the mapping from~$Y'$ to
	$[|X|]$ such that for all $x,y \in V(Y')$ we have
	\mbox{$\tau(x)<\tau(y)\iff \tau'(x)<\tau'(y)$}. Then
	$\mathfrak Y'=(Y',Z_{R_1}',\dots,Z_{R_k}',\tau')$ is a ranked
	twin-model of~$\mathbf A[X]$, whose width is not larger than the one
	of $\mathfrak Y$. \hfill $\vartriangleleft$   
\end{claim}

\begin{lem}
	\label{lem:degTM}
	The Gaifman graph of a ranked twin-model of a $\Sigma$-structure with width
	$d$ is $d+2$-degenerate.
\end{lem}
\begin{proof}
	Let $\mathbf A=(A,R_1(\mathbf A),\dots,R_k(\mathbf A))$ be a
	$\Sigma$-structure, let $\mathfrak T=(Y,Z_{R_1},\dots,Z_{R_k},\tau)$
	be a ranked twin-model of $\mathbf A$ with width $d$, and let $G$ be
	the Gaifman graph of $(Y,Z_{R_1},\dots,Z_{R_k})$. The ranked
	twin-model $\mathfrak T$ (with layers $\mathbf L_i$) defines a
	$d$-se\-quence $\mathbf A_n,\dots,\mathbf A_1$, where $A_i=L_i$ (see
	\Cref{lem:twwtm}).  Let $z$ be the node with $\tau(z)=n-1$ and let
	$u$ and $v$ be its children. Each pair in $Z_{R_i}$ containing $u$
	(except pairs containing  both $u$ and $v$) gives rise (in
	$\mathbf A_{n-1}$) to an $R_i^*$-edge incident to $z$ when
	contracting $u$ and $v$. Thus, the degree of $u$ in $G$ is at most
	$d+2$ ($d$ for the sum of the degrees in the relations $R_i^*$,
	$1$~for the pair~$(u,v)$ in at least one $Z_{R_i}$, and~$1$ for the tree edge
	$(u,z)$). Then, in $G-u$, the vertex $v$ has  degree at most
	$d+1< d+2$. Now note that by removing $u$ and $v$ from~$Y$, and
	redefining $\tau(x)$ as $\min(n-1,\tau(x))$, we get a ranked
	twin-model of $\mathbf A_{n-1}$ (minus $R_i^*$-edges) with width at
	most $d$, whose Gaifman graph is $G-u-v$. By induction, we deduce
	that $G$ is $d+2$-degenerate.
\end{proof}

\section{Full twin-models}
\label{sec:ftm}
To reconstruct a $\Sigma$-structure $\mathbf A$ from a twin-model 
$(Y,Z_{R_1},\dots,Z_{R_k})$, we make use of the tree order $\preceq_Y$
defined by $Y$. As this tree order cannot be obtained as a first-order
transduction of $(Y,Z_{R_1},\dots,Z_{R_k})$ it will be convenient to
introduce a variant of twin-models: the \emph{full twin-model}
associated to a twin-model $(Y,Z_{R_1},\dots,Z_{R_k})$ is the
structure $(V(Y),\prec_Y,Z_{R_1},\dots,Z_{R_k})$.  
\begin{defi}[Transduction $\mathsf S$ and the class $\mathscr F$]
	\label{def:tS}
	The transduction $\mathsf S$ is
	the simple interpretation of $\Sigma$-structures in full twin-models defined
	by formulas
	\begin{align*}
		\rho_0(x)&:=\neg (\exists y\ y \succ_Y x);\\
		\rho_{R_i}(x,y)&:=\exists u\,\exists v\ (u\preceq_Y x)\wedge(v\preceq_Y y)\wedge Z_{R_i}(u,v).
	\end{align*}
	
	$\mathscr F$ is the class of all the full twin-models corresponding to the twin-models in $\mathscr T$.
\end{defi}

The following lemma follows directly from the definition of a
twin-model.
\begin{lem}
	\label{cl:S}
	The class $\mathscr C_0$ is a $2$-bounded $\mathsf S$-transduction of the class $\mathscr F$.	
\end{lem}
\begin{proof}
	For all $\mathbf A\in\mathscr F$,
	if $\mathbf T=(T,\prec,Z_{R_1},\dots,Z_{R_k})$ is a full twin-model
	of  $\mathbf A$ then $\mathsf S(\mathbf T)=\mathbf A$ and
	$|T|=2|A|-1$.
\end{proof}

\begin{lem}
	\label{lem:twwftm}
	Let $\mathfrak T=(Y,Z_{R_1},\dots,Z_{R_k},\tau)$ 
	be  a ranked twin-model with associated 
	full twin-model
	$\mathbf T=(V(Y),\prec_Y, Z_{R_1},\dots,Z_{R_k})$.
	Then the twin-width of\, $\mathbf T$ is at most twice the width of the ranked twin-model 
	$(V(Y),\prec_Y,Z_{R_1},\dots,Z_{R_k},\tau)$.
\end{lem}

\begin{proof}
	Let $I_0,I_1$ be copies of $I(Y)$ and let $p_i:I(Y)\rightarrow I_i$ be
	the ``identity'' for $i=0,1$.  We define the binary rooted tree
	$\widehat Y$ with vertex set $V(\widehat Y)=V(Y) \cup I_1 \cup I_0$,
	leaf set $L(\widehat Y) = V(Y)$, root $r(\widehat Y)=p_o(r(Y))$, and
	parent function
	
	\[ \pi_{\widehat Y}(x) =\begin{cases}
		p_1\circ \pi_Y(x)&\text{if }x\in L(Y)\\
		p_0(x)&\text{if }x\in I(Y)\\
		p_0\circ p_1^{-1}(x)&\text{if }x\in I_1\\
		p_1 \circ \pi_Y\circ p_0^{-1}(x)&\text{if }x\in I_0\setminus \{r(\widehat Y)\}\\
		x&\text{if }x=r(\widehat Y)
	\end{cases}\]
	
	An informal description of $\widehat Y$ is that it is
	obtained by replacing every internal node $v$ of~$Y$ by a
	\emph{cherry} $C_v$ (i.e., a complete binary tree on three vertices)
	such that one leaf of~$C_v$ remains a leaf in $\widehat Y$, the other
	leaf of $C_v$ is linked to the ``children of $v$'', while the root of
	$C_v$ is linked to the ``parent of $v$'' (provided $v$ is not the root
	of $Y$). 
	
	We further define $\widehat Z_\prec=\{(v, p_1(v)): v\in I(Y)\}$ and keep the relations $Z_{R_i}$ as they were defined on $\mathfrak T$. See \Cref{fig:TMTM} for an example.

	\begin{figure}[h!t]
		\centering\includegraphics[width=\textwidth]{TMTM}
		\caption{Construction of the twin-model of a twin-model. Each internal vertex of $Y$ defines three vertices of $\widehat Y$. Here $I=\{\alpha,\beta,\dots\}, I_0=\{\alpha_0,\beta_0,\dots\}$, and $I_1=\{\alpha_1,\beta_1,\dots\}$.
		Purple (heavy) edges are in $Z_R$, green (dotted) edges are in $\widehat Z_\prec$. Note that in the twin-model of the twin-model, only leaves are adjacent by edges in $Z_R$.
		}
		\label{fig:TMTM}
	\end{figure}

	\begin{claim}
		$\widehat{\mathbf T}=(\widehat Y,\widehat Z_\prec, Z_{R_1},\dots,Z_{R_k})$ is a
		twin-model of\, $\mathbf T$.
	\end{claim}
	\begin{claimproof}
		We have $V(Y)=L(\widehat Y)$.  The minimality conditions are
		obviously satisfied for $\widehat Z_\prec$ and~$Z_{R_i}$.  Let
		$\widehat Z=\widehat Z_\prec\cup\bigcup_i Z_{R_i}$.  Consider a
		cycle $\widehat\gamma$ in $\widehat Y\cup \widehat Z$, with all the
		edges in $\widehat Y$ oriented away from the root. Assume for
		contradiction that no two edges in $\widehat Z$ are consecutive in
		$\widehat\gamma$. Then either $\widehat\gamma$ contains a directed
		path of $\widehat Y$ linking to vertices in $L(Y)$ or a directed
		path of $\widehat Y$ linking a vertex in $I(Y)$ to a distinct vertex
		in $V(Y)$. As no such directed paths exist in $\widehat Y$ we are
		led to a contradiction. Hence,
		$\widehat{\mathbf T}$ satisfies the
		consistency condition.

		In order to complete our proof, we still need to prove that
		$\widehat{\mathbf T}$ is indeed a twin-model of~$\mathbf T$, that is that
		 for every $u,v\in V(Y)$ we have (for all $1\leq i\leq k$)
		\begin{align*}
			(u,v)\in Z_{R_i}  \text{~(in ${\mathbf T}$)}\quad&\iff\quad
				\exists u',v'\in V(\widehat Y)\quad (u'\preceq_{\widehat Y}u)\wedge (v'\preceq_{\widehat{Y}}v)\wedge (u',v')\in Z_{R_i}  \text{~(in $\widehat{\mathbf T}$)}\\
			u\prec_Y v  \text{~(in ${\mathbf T}$)}\quad&\iff\quad
			\exists u',v'\in V(\widehat Y)\quad (u'\preceq_{\widehat Y}u)\wedge (v'\preceq_{\widehat{Y}}v)\wedge (u',v')\in \widehat{Z}_{\prec}  \text{~(in $\widehat{\mathbf T}$)}\\
		\end{align*}
		
		Let $(u,v)\in V(Y)$ and let $1\leq i\leq k$..
		
		As in $\widehat{\mathbf T}$ the edges in $Z_{R_i}$ link only leaves of $\widehat{Y}$, we infer that
		there exists $u'\preceq_{\widehat Y}u$ and
		$v'\preceq_{\widehat Y}v$ with $(u',v')\in Z_{R_i}$ (in $\widehat{\mathbf T}$)
		if and only if
		$(u,v)\in Z_{R_i}$ (in $\mathbf T$); see \Cref{fig:TMTM}.

		As 	 in $\widehat{\mathbf T}$ the edges in $\widehat Z_\prec$ only link some leaf $x$ of $\widehat{Y}$ to $p_1(x)$, we infer that  there exists
		$u'\preceq_{\widehat Y}u$ and $v'\preceq_{\widehat Y}v$ with
		$(u',v')\in\widehat Z_\prec$ (in $\widehat{\mathbf T}$) if and only if (up to exchanging $u$ and $v$) we have  $u'=u, v'=p_1(u)$, and
		$v'\preceq_{\widehat Y}v$, that is, if and only if~$u\prec_Y
		v$ (in $\mathbf T$). 
		
		Hence, $\widehat{\mathbf T}$ is a twin-model of $\mathbf T$.
		\vspace{-4mm}
	\end{claimproof}
	
	\medskip
	Let $n=|L(Y)|$.  The next claim shows that we have much freedom in
	defining a ranking for
	$\widehat{\mathbf T}$.
	\begin{claim}
		If $\hat\tau:V(\widehat Y)\rightarrow[2n-1]$ satisfy the labeling
		and monotonicity conditions, then~$\hat\tau$ is a ranking of\,
		$\widehat{\mathbf T}$.
	\end{claim}
	\begin{claimproof}
		Assume $(u,v)\in \widehat Z_\prec$. Then
		$\pi_{\widehat Y}(u)=\pi_{\widehat Y}(v)$ hence the synchronicity
		for~$\widehat Z_\prec$ follows from monotonicity. Assume
		$(u,v)\in Z_{R_i}$. Then $\hat\tau(u)=\hat\tau(v)=2n-1$ hence the
		synchronicity obviously holds.
	\end{claimproof}
	
	\smallskip
	We now define $\hat\tau:V(\widehat Y)\rightarrow [2n-1]$ as follows:
	order the vertices $v\in I_1$ by increasing $\tau\circ
	p_1^{-1}(v)$. For each $v\in I_1$, insert the children of $v$ in
	$I_0$ just after $v$, then add $r(\widehat Y)$ in the very
	beginning. Numbering the vertices of $I_0\cup I_1$ according to this
	order defines $\hat\tau$ on~$I(\widehat Y)$. We extend this function
	to the whole $V(\widehat Y)$ by defining $\hat\tau(v)=2n-1$ for all
	$v\in L(\widehat Y)$.  By construction, the labeling and
	monotonicity properties hold hence $\hat\tau$ is a ranking of
	$\widehat{\mathbf T}$.
	
	Consider a time $1<\hat t<2n-1$ and let $v$ be the vertex with
	$\hat\tau(v)=\hat t$. 
	
		Assume for contradiction that some edge $(x,y)$ belongs to $Z_\prec^\ast$. Then, up to exchanging~$x$ and $y$, $x$ is in $I$ (hence a leaf of $\widehat Y$) and  $p(x)$ is a descendant of $y$ (i.e. $p(x)\succ_{\widehat Y} y$). Thus, $y\preceq_{\widehat Y}  \pi_{\widehat Y}(p(x))=\pi_{\widehat Y}(x)$, contradicting the necessary condition that $x$ and $y$ are non-comparable in $\preceq_{\widehat Y}$. Thus,
	the degree for $Z_\prec^\ast$ is null. 
	
	If $v\in I_1$ we define
	$t=\tau\circ p_1^{-1}(v)$. Then
	$\partial_{\hat t}\widehat Y=p_1(\partial_t Y)$ and the degree for
	$Z_{R_i}^\ast$ in the layer of $\widehat Y$ at time~$\hat t$ is at
	most the degree for $Z_{R_i}^\ast$ in the layer of $Y$ at time $t$.
	
	If $v\in I_0$ we define
	$t=\tau\circ \pi_Y\circ p_0^{-1}(v)$.  Then
	$\partial_{\hat t}\widehat Y$ is $p_1(\partial_t Y)$ in which we
	remove the parent of $v$ and add $v$ and (maybe) the sibling of
	$v$. Compared to the layer at time
	$\hat\tau\circ \pi_1(\pi_Y\circ p_0^{-1}(v))$, the red degree can
	only increase because some relations $Z_{R_i}^\ast$ from a vertex $u$ are
	incident to $v$ and its sibling. It follows that the maximum
	$Z_{R_i}^\ast$ is at most doubled.
\end{proof}

From what precedes, we deduce the following theorem, which may be of independent interest.

\begin{thm}
	\label{thm:ftm}
	Let $\Sigma$ be a binary signature.
	Every binary $\Sigma$-structure with twin-width $t$ has a full twin-width model $\mathbf T=(V(Y),\prec_Y, Z_{R_1},\dots,Z_{R_k})$
	with twin-width at most $2t$, associated to a 
	ranked twin-width model $\mathfrak T=(V(Y),\prec_Y,Z_{R_1},\dots,Z_{R_k},\tau)$ with width $t$ and $d+2$-degenerate Gaifman graph.
\end{thm}
\begin{proof}
	The existence of a ranked twin-width model $\mathfrak T$ with width $t$ follows from \Cref{lem:twwtm}. According to \Cref{lem:degTM},  $\mathfrak T$ has a $d+2$-degenerate Gaifman graph and, according to \Cref{lem:twwftm}, the associated full twin-width model $\mathbf T$ has twin-width at most $2t$.
\end{proof}

\section{Ordered twin-models}

Recall that a \emph{preordering} of the vertices of a rooted tree is the discovery order of the vertices of a (depth-first search) traversal of the tree starting at its root.

	\begin{lem}
	\label{lem:otm}
	Let $\mathscr O$ be the class of binary tree orders and let
	$\mathscr Y^<$ be the class of rooted binary trees, with vertices ordered by some
	preordering. Then there exist simple  transductions~$\mathsf L$ and~$\mathsf O$ such that $(\mathsf L,\mathsf O)$ is a transduction
	pairing of $\mathscr O$ and $\mathscr Y^<$.
	\end{lem}

\begin{proof}
	We define two simple transductions. The first transduction maps
	binary tree orders~$\prec$ into the preorder defined by some
	traversal of $Y$.
	
	$\mathsf{L}$ is defined as follows: we consider a mark $M$ on the
	vertices and define
	\vspace{-1mm}
	\begin{align*}	
		\rho_E(x,y)&:=\bigl((x\prec y)\wedge\ \forall v\ \neg((x\prec v)\wedge (v\prec y))\bigr)\vee 
		\bigl((y\prec x)\wedge\ \forall v\ \neg((y\prec v)\wedge (v\prec x))\bigr)\\
		\rho_<(x,y)&:=(x\prec y)\vee\ \neg(y\preceq x)\wedge
		\exists u\,\exists v\,\exists w\, \Bigl(\forall z
		\bigl((w\prec z)\rightarrow \neg((z\prec u)\wedge (z\prec v))\bigr)\wedge \\[-1mm]
		& \hspace{9mm}
		(u\preceq x)\wedge (v\preceq y)\wedge (w\prec u)\wedge (w\prec v)\wedge \rho_E(u,w)\wedge \rho_E(v,w)\wedge M(u)\Bigr). 
	\end{align*} 
	\vspace{-4mm}
	
	Consider a binary tree order $(V(Y),\prec)\in\mathscr O$, and let $Y$
	be the rooted binary tree defined by $\prec$.  Recall that the
	\emph{preordering} of $Y$ defined by some plane embedding of $Y$ (that is
	to an ordering, for each node $v$, of the children of $v$) is a linear
	order on $V(Y)$ such that for every internal node $v$ of $Y$, one
	finds in the ordering the vertex~$v$, then the first child of~$v$
	and its descendants, then the second child of $v$ and its
	descendants.  Let $<$ be the preordering defined by some plane embedding
	of $Y$.  (Note, that, despite its name, this is a total order.)
	 Let us mark by $M$ all the nodes of $Y$ that are the first
	child of their parent.  The formula $\rho_E$ defines the cover
	graph of $\prec$, thus~$E$ is the adjacency relation of $Y$. Let $x,y$
	be nodes of $Y$. If $x=y$, then $\rho_<(x,y)$ does not hold. If~$x$ and
	$y$ are comparable in $\prec$, then $\rho_<(x,y)$ is equivalent to
	$x\prec y$. Otherwise, let~$w$ be the infimum of $x$ and $y$, and let
	$u$ and $v$ be the children of~$w$ such that $u\prec x$ and
	$v\prec y$. Then~$\rho_<(x,y)$ holds if $u$ is the first child of
	$w$, that is, if $u$ is marked. Altogether, we have~$\rho_<(x,y)$ if
	and only if $x<y$. Hence, $(V(Y),\prec)\in\mathsf L(Y^<)$, where $Y^<$
	is~$Y$ ordered by $<$.
	
	The transduction $\mathsf O$ is defined as follows:
	\[
	\rho_\prec(x,y) :=(x<y)\wedge \forall z\,\forall w\ 
	((x< z)\wedge (z\leq y)\wedge E(z,w))
	\rightarrow (x\leq w).
	\]
	
	Let $Y^<$ be a rooted binary tree $Y$ with preorder $<$ and let
	$\prec$ be the corresponding tree order. If $x\geq y$, then
	$\rho_\prec(x,y)$ does not hold, so we assume $x<y$.  Assume $x$ is an
	ancestor of $y$ in~$Y$, then all the vertices $z$ between $x$ and $y$
	in the preorder are descendants of $x$ thus any neighbor of these are
	either descendants of $x$ or $x$ itself thus $\rho_\prec(x,y)$
	holds. Otherwise, let~$w$ be the infimum of $x$ and $y$ in $Y$ and let
	$z$ be the child of $w$ that is an ancestor of $y$. Then $z$ is
	between $x$ and $y$ is the preorder, $z$ is adjacent to $w$, but $w$
	appears before $x$ in the preorder. Thus, $\rho_\prec(x,y)$ does not
	hold.  It follows that $\rho_\prec(x,y)$ is equivalent to $x\prec y$
	thus $Y^<\in\mathsf O(V(Y),\prec)$.  Thus, $(\mathsf L,\mathsf O)$ is a
	transduction pairing of $\mathscr O$ and $\mathscr Y^<$.
\end{proof}

\begin{defi}[The classes $\mathscr O_0,\mathscr Y_0^<$, and $\mathscr T^<$]
	\label{def:clOY}
	The class~$\mathscr O_0$ is the reduct of the class $\mathscr F$, obtained by keeping only the tree order relation; the class $\mathscr Y_0^<$ is the class of all rooted binary trees corresponding to the tree orders in $\mathscr O_0$ ordered by some preordering, so that $(\mathsf L,\mathsf O)$ is a transduction
	pairing of $\mathscr O_0$ and $\mathscr Y_0^<$.
	
	The class $\mathscr T^<$ is the class of ordered
	twin-models obtained from the twin-models in $\mathscr T$ by adding
	a linear order defined by some preordering of the rooted tree of the tree
	model.
\end{defi}
Note that
the natural generalization $(\widehat{\mathsf L},\widehat{\mathsf O})$ of 
$(\mathsf L,\mathsf O)$ is a transduction pairing of $\mathscr F$ and~$\mathscr T^<$.

\begin{defi}[The class $\mathscr G$]
	\label{def:clG}
	The class $\mathscr G$ is the class of the Gaifman graphs of the structures in $\mathscr T$.
\end{defi}

\begin{lem}
	\label{lem:toGaifman}
	Let $\mathsf G=\mathsf{Gaifman}_\Sigma$. 
	There exists a
	simple transduction $\mathsf{U}$ such that
	$(\mathsf{G},\mathsf{U})$ 
	is a transduction pairing of $\mathscr T$ and~$\mathscr G$.
\end{lem}
\begin{proof}
	According to \Cref{lem:twwftm} the class $\mathscr F$
	has bounded twin-width. 
	According to \Cref{thm:trans}, as~$\mathscr T^<$ is an $\widehat{\mathsf{L}}$-trans\-duction
	of~$\mathscr F$ it has bounded twin-width. Thus, the class~$\mathscr T$, being a reduct of $\mathscr T^<$, has bounded
	twin-width. It follows from \Cref{lem:degTM} that the class~$\mathsf{G}(\mathscr T)$ is degenerate, hence, according to
	\Cref{thm:sptww}, it has bounded expansion and, according to
	\Cref{thm:BEchist}, bounded star chromatic number (at most $c$).  It follows
	from \Cref{lem:Gaifman} that, defining $\mathsf U=\mathsf{Unfold}_{\Sigma,c}$,
	$(\mathsf{G},\mathsf{U})$ 
	is a transduction pairing of $\mathscr T$ and \mbox{$\mathsf{G}(\mathscr T)=\mathscr G$}.
\end{proof}

\begin{defi}[The class $\mathscr G^<$]
	\label{def:clGo}
	The class $\mathscr G^<$ is the class of ordered graphs obtained from the structures in $\mathscr T^<$ by applying the natural generalization $\widehat{\mathsf G}$ of the interpretation $\mathsf{Gaifman}_\Sigma$.
\end{defi}
Note that, denoting by $\widehat{\mathsf U}$ the natural generalization of the transduction $\mathsf U$,  it follows from \Cref{lem:toGaifman} that 
$(\widehat{\mathsf G},\widehat{\mathsf U})$ is a transduction pairing of $\mathscr T^<$ and $\mathscr G^<$.

\section{Permutations and the main result}
\label{sec:perm}

When we speak about transductions of permutations, we consider the
permutations as defined in \Cref{sec:rel}. Hence, the language used to
define the transduction can use the binary relations $<_1,<_2$, as
well as equality. 
\begin{lem}
	\label{cl:sparse2permutation} Let $c\in\mathbb N$ and
	let~$\mathscr S^<$ be a class of ordered graphs, and let~$\mathscr S$ be the reduct of~$\mathscr S^<$ obtained by forgetting the linear order.
	
	Assume that the graphs in $\mathscr S$ have star chromatic number at most~$c$. 
	Then, there exist a copying transduction~$\mathsf T_1$ with  \mbox{$\bl(\mathsf T_1)=c+1$}, a simple transduction~$\mathsf T_2$, and a class
	$\mathscr P$ of  permutations such that 
	$(\mathsf T_1,\mathsf T_2)$ is a transduction pairing of~$\mathscr S^<$ and~$\mathscr P$.
\end{lem}

We give here an informal description of the transductions~$\mathsf T_1$ and $\mathsf T_2$ and refer to \Cref{fig:Ex-perm} for an example:
The transduction $\mathsf T_1$ is used to compute a permutation $\sigma$ from an ordered  graph $G^<$ and works as follows: we first compute a star coloring $\gamma$ of $G$ with $c$ colors and orient edges so that bicolored stars are oriented from their centers. Then we blow each vertex into $(u,1),\dots,(u,c+1)$. From this we keep only the vertices of the form $(v,c+1)$ or of the form $(v,i)$ if $v$ has an in-neighbor colored $i$. The linear order $<_1$ orders pairs~$(u,i)$ by first coordinate first (using $<$) then by increasing $i$. The linear order $<_2$ is a succession of intervals ending with a vertex of the form $(v,c+1)$ (these intervals being ordered according to the order on $v$); the interval ending with $(v,c+1)$ contains the pairs $(u,\gamma(v))$, for all the out-neighbors~$u$ of $v$, ordered by first coordinate.
The transduction $\mathsf T_2$ is used to compute an ordered graph $G^<$ from a permutation $\sigma$. It works as follows. First, we mark some elements. These elements will correspond to the vertices of~$G^<$, and the linear order $<$ will be the restriction of $<_1$ to these elements. Each vertex $v$ of~$G^<$ defines a maximal interval~$A(v)$ in~$<_1$ ending with~$v$ and containing no other marked element, and a maximal interval $B(v)$ in~$<_2$ ending with $v$ and containing no other marked element.
In $G^<$, a vertex $u$ is adjacent to a vertex $v\neq u$ if $A(u)$ intersects $B(v)$ or $A(v)$ intersects $B(u)$.

\begin{figure*}[h!t]
	\centering
	\includegraphics[width=.95\textwidth]{Ex-perm}	
	
	\medskip
	Thus, the permutation obtained for this example is
	\vspace{-3mm}
	
	\setcounter{MaxMatrixCols}{33}
	\[\sigma=\begin{pNiceMatrix}[small,columns-width = 4pt]
		1&2&3&4&5&6&7&8&9&10&11&12&13&14&15&16&17&18&19&20&21&22&23&24&25&26&27&28&29&30&31&32&33\\
		3&2&6&5&9&25&8&31&11&7&10&14&29&12&4&15&1&18&17&20&16&19&23&21&26&13&22&24&32&28&27&30&33\\	
	\end{pNiceMatrix}\]
	\medskip
	
	\includegraphics[width=.95\textwidth]{Perm-ex}	
	
	\caption{In the top, the transduction $\mathsf T_1$ (we assume $c=5$). 
		In the bottom, the transduction $\mathsf T_2$. In gray, the marked elements of $\sigma$, which are the vertices of $G^<$. In $G^<$,
		an edge links $15$ ($g$ in the top) and $28$ ($v_6$ in the top) as $A(15)$ intersects $B(28)$.
		(The linear order $<$ is the left-hand traversal preorder of the rooted tree.)
	}
	\label{fig:Ex-perm}
\end{figure*}

\begin{proof}
	Let $G\in\mathscr S$, let $\gamma:V(G)\rightarrow[c]$ be a star coloring of $G$, and let $\vec G$ be an orientation of $G$ obtained by orienting all bicolored stars from their roots.
	We mark a vertex $v\in V(G)$ by $M_i$ if $\gamma(v)=i$ and, for $I\subseteq [c]$, by 
	$N_I$ if $I$ is the set of the $\gamma$-colors of the in-neighbors of $v$.
	Let $\mathsf C_{c+1}$ be the $(c+1)$-blowing transduction. The vertices of $\mathsf C_{c+1}(G)$ are the pairs $(u,i)\in V(G)\times[c+1]$, and there are new predicates $P_j$ (with $j\in[c+1]$), where $P_j(x)$ holds if $x$ is of the form $(u,j)$, for some $u\in V(G)$.
	We define $\rho(x):=P_{c+1}(x)\vee \bigvee_{I\subseteq[c]}\bigvee_{i\in I}(N_I(x)\wedge P_i(x))$. (Note that if $x$ is $(u,j)$ then $N_I(x)$ is $N_I(u)$.)
	Hence $W=\rho(\mathsf C_{c+1}(G))$ is  the union of  $V(G)\times\{c+1\}$ and the set of all pairs $(u,i)\in V(G)\times [c]$ such that $u$ has an in-neighbor in $\vec G$ with color $i$. 
	We consider the subgraph $H$ of $\mathsf C_{c+1}(G)$ induced by $W$.
	For $x\in V(H)$ we define $f_{c+1}(x)$ as $x$ if $P_{c+1}(x)$ or as the (only) $\sim$-neighbor $y$ of $x$ with
	$P_{c+1}(y)$, that is:
	\[
	f_{c+1}(x)=y\quad\iff\quad P_{c+1}(y)\wedge((x=y)\vee (x\sim y)).
	\]
	If $x$ is of the form $(u,i)$ then $f_{c+1}(x)$ is $(u,c+1)$. 
	Then, for $i\in I\subseteq [c]$ and whenever $N_I(x)$ holds, we define $f_i(x)$ as the (only) neighbor $y$ of $f_{c+1}(x)$  in $P_{c+1}$ and  $M_i$. Hence,
	\[
	f_i(x)=y\quad\iff\quad E(f_{c+1}(x),y)\wedge P_{c+1}(y)\wedge M_i(y).
	\]
	We now define the linear orders $<_1$ and $<_2$ as follows (See \Cref{fig:Ex-perm} for an illustration.)
	\begin{align*}
	x<_1y\qquad&\iff\qquad
	\begin{cases}
		\text{either } f_{c+1}(x)<f_{c+1}(y),\\
		\text{or }f_{c+1}(x)=f_{c+1}(y), P_i(x), P_j(y),\text{ and }i<j.
	\end{cases}\\
	x<_2 y\qquad&\iff\qquad
	P_i(x), P_j(y),\text{ and }
	\begin{cases}
\text{either }f_i(x)<f_j(y),\\ 
\text{or }f_i(x)=f_j(y) \text{ and } f_{c+1}(x)<f_{c+1}(y).
\end{cases}
\end{align*}
(Note that, in the last condition, $f_i(x)=f_j(y)$ implies  $i=j$.)

	We call the obtained permutation $\sigma(G^<)$. 
	Note that this permutation depends on some arbitrary choices of star coloring. 
	We further define $\mathscr P=\{\sigma(G^<): G^<\in\mathscr S^<\}$.
	The transduction~$\mathsf T_1$ is defined as the composition of $\mathsf C_{c+1}$, the interpretation reducing the domain to the vertices satisfying~$\rho(x)$, then the interpretation defining $<_1$, $<_2$ and forgetting all the other relations.  Hence $\sigma(G^<)\in\mathsf T_1(G^<)$.

	The definition of $\mathsf T_2$ is as follows: we consider a predicate~$M$
	in such a way that the maximum element of $<_1$
	is in $M$. The domain $V$ of $\mathsf T_2(\sigma)$ is $M(\sigma)$.
	The linear order~$<$ is the restriction of $<_1$ to
	$V$. Then, $x$ is adjacent to $y$ if there exists $z\notin V$ and $(i,j)\in\{(1,2),(2,1)\}$ with $z<_i x$, $z<_j y$, and no vertex in~$V$ is between~$z$ and $x$ in $<_i$ and no vertex in $V$ between~$z$ and $y$ in $<_j$. It is easily checked that for every $G^<\in\mathscr S^<$ we have $G^<\in\mathsf T_2(\sigma(G))$
	(see \Cref{fig:Ex-perm}).
	According to \Cref{lem:pairing} it follows that $(\mathsf T_1,\mathsf T_2)$ is a transduction pairing of~$(\mathscr S^<,\mathscr P)$.
\end{proof}

\begin{thm}
	\label{thm:main}
	For every class $\mathscr C_0$ of binary structures with
	twin-width at most $t$ there exists a proper permutation class 
	$\overline{\mathscr P}$, an integer $k$, and a transduction $\mathsf T$, such that 
	$\mathscr C_0$ is a $k$-bounded $\mathsf T$-transduction of $\overline{\mathscr P}$.
	Precisely,   for every graph
	$G\in\mathscr C_0$ there is a permutation $\sigma\in\overline{\mathscr P}$ on at most 
	$k|G|$ elements with $G\in\mathsf T(\sigma)$. 
\end{thm}
\begin{proof}
	Let $\mathscr C_0$ be a class of binary structures with
	twin-width at most $t$. Let $\mathscr T$ be a class of twin-models
	obtained by optimal contraction sequences of graphs in $\mathscr C_0$,
	and let $\mathscr F$ be the class of the corresponding full
	twin-models. According to \Cref{lem:twwftm} $\mathscr F$ has
	twin-width at most~$2t$, moreover, applying the transduction
	$\mathsf L$ on $\prec$ we transform $\mathscr F$ into the 
	class~$\mathscr T^<$, whose reduct is $\mathscr T$ (see \Cref{lem:otm}). 
	Let $\mathscr G^<$ be
	the class obtained from $\mathscr T^<$ be taking the Gaifman graphs
	of the relations distinct from the linear order, and keeping the
	linear order, and let $\mathscr G$ be the reduct of $\mathscr G^<$
	obtained by forgetting the linear order. Thus
	$\mathscr G=\mathsf{Gaifman}(\mathscr T)$. As the classes
	$\mathscr G^<$ and its reduct $\mathscr G$ are transductions of the
	class $\mathscr F$ they have bounded twin-width, by~\Cref{thm:trans}. Moreover, the class
	$\mathscr G$ is degenerate hence it has bounded expansion and, in
	particular, bounded star chromatic number.
	 It follows that we have a transduction pairing of 
	$\mathscr G^<$ and a class $\mathscr P$ of permutations.
	As the class of all finite graphs has unbounded twin-width and is a transduction of the class of all permutations, the class of all permutations has unbounded twin-width. Thus, as $\mathscr P$ has bounded twin-width, it is a proper class of permutations.
	From the transduction pairing of $\mathscr T^<$ and $\mathscr G^<$ and the one of~$\mathscr F$ and $\mathscr T^<$ we deduce that there is a transduction pairing of~$\mathscr F$ and~$\mathscr P$. As $\mathscr C_0$ is a transduction of~$\mathscr F$ we conclude
	that~$\mathscr C_0$ is a transduction of $\mathscr P$.
	
	Note that the class $\mathscr C_0$ is obviously also a $\mathsf T$-transduction of the permutation class $\overline{\mathscr P}$ obtained by closing $\mathscr P$ under sub-permutations. 
\end{proof}

\begin{cor} \label{cor:small}
	Every class of graphs with bounded
	twin-width contains at most $c^n$ non-isomorphic graphs on $n$
	vertices (for some constant $c$ depending on the class).
\end{cor} 
\begin{proof}
	Let $\mathscr C_0$ be a class with bounded twin-width. As twin-width is monotone with respect to induced subgraph inclusion, we may assume that $\mathscr C_0$ is hereditary.
	According to \Cref{thm:main}, there exists a proper permutation class 
	$\overline{\mathscr P}$, an integer~$k$, and a transduction $\mathsf T$, such that for every $G\in\mathscr C_0$ there is a permutation $\sigma\in\overline{\mathscr P}$ on $k|G|$ elements with $G\in\mathsf T(G)$. Let $m$ be the number of unary predicates used by the transduction.
	According to the Marcus-Tardos theorem~\cite{MarcusT04}, for every proper permutation $\mathscr P$ there exists a constant $a$ such that~$\mathscr P$ contains at most $a^n$ permutations on $n$ elements. For each permutation on $n$ elements, there are $2^{mn}$ possible choices for the interpretation of the $m$ predicates (as each predicate defines a subset of elements). 
	It follows that $\mathscr C_0$ contains
	at most $\sum_{i=1}^{kn}a^i\,2^{mi}=O((a^k2^{mk})^n)$ non-isomorphic graphs with at most $n$ vertices. Thus, there exists a constant~$c$ such that $\mathscr C_0$ contains at most~$c^n$ non-isomorphic graphs with $n$ vertices.
\end{proof}

\vspace{-3mm}
\section{Further remarks}
It was proved by Bonnet \emph{et al.}~\cite{twin-width4} that a class of graphs $\mathscr C$ has bounded twin-width if and only if it is the reduct of a monadically dependent class of ordered graphs. This implies  the following duality type statement for every class $\mathscr C^<$ of ordered graphs:

\vspace{-5mm}
\begin{align*}
	\exists \text{ permutation }\sigma & \text{ with }  \xy\xymatrix{{\rm Av}(\sigma)\ar@{->>}[r]&\mathscr C^<}\endxy\\[-1mm]
	& \hspace{-1mm}\iff\hspace{2.5mm}\xy\xymatrix{\mathscr C^<\ar@{->>}|{/}[r]& \mathscr U}\endxy,
\end{align*}
where ${\rm Av}(\sigma)$ denotes the class of all permutations avoiding the pattern $\sigma$, $\mathscr U$ denotes the class of all graphs,  $\xy\xymatrix{\ar@{->>}[r]&}\endxy$  means the existence of a transduction, and $\xy\xymatrix{\ar@{->>}|{/}[r]&}\endxy$  means the non-existence of a transduction.

Interestingly, transductions relate some classical classes of graphs to well studied permutation classes.
\begin{prop}
	\label{prop:ex}
A class $\mathscr C$ of graphs has
\begin{itemize}
	\item bounded linear clique-width if and only if it is a transduction of the class ${\rm Av}(21)$ of identities;
	\item bounded clique-width if and only if it is a transduction of 
	the class ${\rm Av}(231)$ of stack-sortable permutations (or, equivalently, if and only if it is a transduction of the class ${\rm Av}(2413,3142)$ of separable permutations).
\end{itemize}
\end{prop}
\begin{proof}
	It is proved in \cite{Colcombet07} that a class of graphs $\mathscr C$ has bounded linear clique-width if and only if it is a transduction of the class of  linear orders, that is: if and only if it is a transduction of the class of identities.
	
	It is also proved in \cite{Colcombet07} that a class of graphs $\mathscr C$ has bounded clique-width if and only if it is a transduction of the class of (binary) tree orders.
	The first equivalence of the second item follows from the next claim.
	\begin{claim}
		Binary tree-orders are transduction equivalent to stack-sortable (i.e.\ $231$-avoiding) permutations.
	\end{claim}
	\begin{claimproof}
		From binary tree-order to $231$-avoiding permutations.
		We let $<_1$ to be the lexicographic order (using two marks) and we define $<_2$
		as follows $x<_2 y$ if either $x$ is on the left side at $x\wedge y$ or $x<y$  and the successor
		$x'$ of $x$ with $x'\leq y$ is on the left side, or $x>y$ and the successor $y'$ of $y$ with $y'\leq x$ is on the right side.
		
		Conversely, from a $231$-avoiding permutation $(X,<_1,<_2)$ we define the tree order by $u\wedge y$ is the $<_1$-maximum element $x$ such that $x<_1 u, x<_1 v$, and $x$ is between $u$ and $v$ in~$<_2$.
	\end{claimproof}
Let us sketch  the second equivalence (with separable permutations):
One direction is obvious, as stack-sortable permutations are separable. For the other direction, it is easily shown that separable permutations are  transductions of the tree order defined by their separation tree. 
\end{proof}

Proposition~\ref{prop:ex} shows that two permutations classes (like stack-sortable and separable permutations) may be transduction equivalent but not Wilf-equivalent. To the opposite, it is not difficult to prove that, though they are Wilf-equivalent,  the class of $321$-avoiding permutations is not a transduction of the class of stack-sortable avoiding permutations. (This follows from the fact that  the class of grids is a transduction of ${\rm Av}(321)$ but has unbounded clique-width.)

The connection between classes of ordered graphs and permutation classes might well be even deeper than what is proved in this paper. 

\begin{conj}
	Every hereditary class $\mathscr C^<$ of ordered graphs is transduction equivalent to a permutation class.
\end{conj}

This conjecture is known to hold  if the class $\mathscr C^<$ is not monadically dependent, as it is then transduction equivalent to the class of all permutations \cite{twin-width4}; it also holds if the reduct~$\mathscr C$ of~$\mathscr C^<$ is biclique-free, as either $\mathscr C^<$ is not monadically dependent (previous item), or it has bounded twin-width \cite{twin-width4}. Then, since biclique-free classes of bounded twin-width have bounded expansion~\cite{twin-width2}, and according to \Cref{thm:BEchist} and \Cref{cl:sparse2permutation} the class $\mathscr C^<$ is transduction equivalent to a permutation class; finally, it also holds if the reduct~$\mathscr C$ of $\mathscr C^<$  is a transduction of a class with bounded expansion as $\mathscr C$ is then transduction equivalent to a bounded expansion class $\mathscr D$ \cite{SBE_TOCL} and this transduction equivalence can be extended to a transduction equivalence of $\mathscr C^<$ and an expansion $\mathscr D^<$ of $\mathscr D$, which falls in the previous case.

\section*{Acknowledgments}
The authors wish to express their gratitude to the referees of a first version of this paper, for their careful reading and their most valuable suggestions.

\bibliographystyle{alphaurl}
\bibliography{lmcs}

\end{document}